\pdfoutput=1
\documentclass{article}

\usepackage{kr}

\usepackage{times}
\usepackage{soul}
\usepackage{url}
\usepackage[hidelinks]{hyperref}
\usepackage[utf8]{inputenc}
\usepackage[small]{caption}
\usepackage{graphicx}
\usepackage{amsmath}
\usepackage{amsthm}
\usepackage{booktabs}
\usepackage{algorithm}
\usepackage{algorithmic}
\usepackage{microtype}
\usepackage{tikz}
\usetikzlibrary{arrows}

\newtheorem{lemma}{Lemma}
 
\newtheorem{theorem}{Theorem}

\title{Ontology-Mediated Querying on Databases of Bounded Cliquewidth}

\author{%
$\text{Carsten Lutz}^1$\and
$\text{Leif Sabellek}^2$\and
$\text{Lukas Schulze}^1$ \\
\affiliations
$^1$Department of Computer Science, Leipzig University, Germany \\
$^2$Department of Computer Science, University of Bremen, Germany \\
\emails
\{clu, lschulze\}@informatik.uni-leipzig.de, sabellek@uni-bremen.de
}

\usepackage[T1]{fontenc}
\usepackage{amssymb}
\usepackage{pifont}
\usepackage{misc}
\usepackage{xspace}
\usepackage{thmtools}
\usepackage{xcolor}
\usepackage{enumitem}
\usepackage{bookmark}
\usepackage{tikz}
\usepackage{subcaption}             
\usepackage{pgfplots}               
\usepackage{mathdots}               
\usepackage{stmaryrd}               
\usepackage{etoolbox}               
\usepackage{thm-restate}            

\usetikzlibrary{positioning, decorations, trees, math}

\newtheorem{example}{Example}

\newcommand{\displayskips}{%
  \setlength{\abovedisplayskip}{4pt}%
  \setlength{\belowdisplayskip}{4pt}%
  \setlength{\abovedisplayshortskip}{4pt}%
  \setlength{\belowdisplayshortskip}{4pt}}
\appto{\normalsize}{\displayskips}
\appto{\small}{\displayskips}
\appto{\footnotesize}{\displayskips}

\newcommand{\GFtwo}{\texorpdfstring{$\text{GF}_2$\xspace}{GF2\xspace}}

\newcommand{\blank}{{b\mkern-8.5mu/}}

\begin{document}

\maketitle

\begin{abstract}
  We study the evaluation of ontology-mediated queries (OMQs) on
  databases of bounded cliquewidth from the viewpoint of parameterized
  complexity theory.
  As the ontology language, we consider the description logics \ALC
  and \ALCI as well as the guarded two-variable fragment GF$_2$ of
  first-order logic. Queries are atomic queries (AQs), conjunctive
  queries (CQs), and unions of CQs. All studied OMQ problems are
  fixed-parameter linear (FPL) when the parameter is the size of the
  OMQ plus the cliquewidth. Our main contribution is a detailed analysis of
  the dependence of the running time on the parameter, exhibiting
  several interesting
  effects. 
\end{abstract}

\section{Introduction}

Ontology-mediated querying is an established subfield of knowledge
representation. The general aim is to enrich a database with an
ontology to add domain knowledge and to extend the vocabulary
available for query formulation
\cite{DBLP:journals/tods/BienvenuCLW14,BiOr15,DBLP:conf/rweb/CalvaneseGLLPRR09}. Expressive
description logics (DLs) such as \ALC and \ALCI are a popular choice
for the ontology language as they underpin the widely used OWL DL
profile of the OWL 2 ontology language.  However, ontology-mediated
querying with these languages is {\sc coNP}-complete in data
complexity even for simple atomic queries (AQs) of the form $A(x)$
and thus scales poorly to larger amounts of data. One reaction to this
problem is to resort to approximate query answers, as done for example
in \cite{DBLP:journals/jair/ZhouGNKH15,DBLP:conf/kr/HagaLSW21}. A
potential alternative is to only admit databases from a class that is
sufficiently restricted so that non-approximate answers can be
computed in \PTime, or to decompose the input database into component
databases from such a class.

The aim of this paper is to study ontology-mediated querying on
classes of databases of bounded cliquewidth.  This is relevant because
such classes are maximal known ones on which ontology-mediated
querying with DLs such as \ALC and \ALCI is in \PTime in data
complexity, even when unions of conjunctive queries (UCQs) are
admitted as queries.  This follows from a folklore translation of
ontology-mediated querying into MSO$_1$ and the known result that
MSO$_1$ has \PTime data complexity on structures of bounded
cliquewidth \cite{DBLP:journals/mst/CourcelleMR00}. In fact, it even
has linear time data complexity (when a $k$-expression for
constructing the database is provided). We recall that bounded
treewidth implies bounded cliquewidth since every database of
treewidth~$k$ has cliquewidth at most $3 \cdot 2^{k-1}$
\cite{corneil2005relationship}. The converse is false because
databases of bounded treewidth are sparse in the sense of
graph theory, whereas databases of bounded cliquewidth may be dense.

We do not claim that databases encountered in practical applications
typically have small cliquewidth. In fact, an empirical study of
real-world databases has found that their treewidth tends to be high
\cite{DBLP:conf/icdt/ManiuSJ19}, and one may expect similar results
for cliquewidth. The same study, however, also points out the
opportunities that lie in the decomposition of a database into a part
of high treewidth and parts of low treewidth, and this has in fact
been used successfully to implement efficient querying
\cite{DBLP:conf/sigmod/Wei10,DBLP:conf/edbt/AkibaSK12,DBLP:journals/tods/ManiuCS17}. While
we are not aware that similar approaches based on cliquewidth have yet
been studied, it seems entirely reasonable to pursue them. The work
presented here may provide a foundation for such an
endeavour. 

In this paper, we focus on the framework of parameterized complexity
theory. Linear time data complexity with a uniform algorithm trivially
implies \emph{fixed-parameter linearity (FPL)} when the
parameter is the size of the ontology-mediated query (OMQ). Recall
that FPL is defined like the more familiar \emph{fixed-parameter
  tractability (FPT)} except that the running time may depend only
linearly on the size of the database, that is, it must be
$f(|Q|) \cdot |\Dmc|$ where $f$ is a computable function, $Q$ the OMQ,
and \Dmc the database.  Our main goal is to determine the optimal
running time of FPL algorithms for ontology-mediated querying on
databases of bounded
cliquewidth 
and to provide algorithms that achieve this running time.

As the ontology language, we consider the DLs \ALC and \ALCI and the
more expressive guarded two-variable fragment GF$_2$ of first-order
logic. As queries, we admit AQs, conjunctive queries
(CQs), and UCQs. An \emph{ontology-mediated query
  (OMQ)} combines an ontology and a query, and an \emph{OMQ language}
is determined by the choice of an ontology language and a query
language. For instance, $(\ALC,\text{AQ})$ is the OMQ language that
combines \ALC ontologies with AQs. We find that the running time
may be single or double exponential in $|Q|$, depending on the choice
of the OMQ language. There are several surprising effects, which we
highlight in the following. An overview is provided~in~Figure~\ref{overview}.

We show that in $(\ALCI,\text{AQ})$, evaluating an OMQ $Q$ on a
database \Dmc is possible in time $2^{O(|Q| \cdot k)} \cdot |\Dmc|$
where $k$ is the cliquewidth of \Dmc.\footnote{We explain our use of
   the $O$-notation in the appendix.}
We assume here that a
$k$-expression for constructing \Dmc is given (otherwise, the running
time is cubic in~$|\Dmc|$ and thus FPT, but not FPL).  Note that this
implies FPL even when the parameter is the size of the OMQ plus the
cliquewidth of the database (and so the cliquewidth of the database
needs not be bounded by a constant).  In $(\text{GF}_2,\text{AQ})$, in
contrast, we achieve a running time of
$2^{2^{{O(|Q|)}}\cdot k^2} \cdot |\Dmc|$ and show that attaining
$2^{2^{o(|Q|)}} \cdot \mn{poly}(|\Dmc|)$ is impossible even on
databases of cliquewidth~2 unless the exponential time hypothesis
(ETH) fails. This is interesting for several reasons. First, it is
folklore that OMQ evaluation in $(\text{GF}_2,\text{AQ})$ is
\ExpTime-complete in combined complexity on unrestricted databases,
and thus one has to 'pay' for the running time to be polynomial in
$|\Dmc|$ on databases of bounded cliquewidth by an exponential
increase in the running time in $|Q|$. Second, the higher running time
for $(\text{GF}_2,\text{AQ})$ compared to $(\ALCI,\text{AQ})$ is
\emph{not} due to the fact that GF$_2$ is more expressive than
\ALCI. In fact, the GF$_2$ ontology used in the lower bound proof can
be expressed in \ALC and the increase in complexity is due to the
higher succinctness of~GF$_2$.  And third, we also show that in
$(\text{GF}_2,\text{AQ})$, OMQs can be evaluated in time
$2^{O(|Q| \cdot k^2)} \cdot |\Dmc|$ where $k$ is the
\emph{treewidth} of~\Dmc. Thus, the exponential difference in
complexity between $(\ALCI,\text{AQ})$ and $(\text{GF}_2,\text{AQ})$
that we observe for cliquewidth does not exist for treewidth.

The above results concern AQs, and transitioning to (U)CQs brings
about some interesting differences. We 
show that in $(\ALCI,\text{UCQ})$, evaluation of an OMQ $Q$ is
possible in time
$2^{|\Omc| \cdot 
  k^{O(|q| \log(|q|))}} \cdot |\Dmc|$
where \Omc is the ontology in $Q$ and $q$ is the query in $Q$.  This
is complemented by the result that already in $(\ALC,\text{CQ})$, a
running time of $2^{2^{o(|Q|)}} \cdot \mn{poly}(|\Dmc|)$
cannot be achieved even on databases of
cliquewidth~3 unless ETH fails. This should be contrasted with the
fact that OMQ evaluation in $(\ALC,\text{CQ})$ and $(\ALC,\text{UCQ})$
is \ExpTime-complete in combined complexity on unrestricted databases
\cite{DBLP:conf/dlog/Lutz08} and thus also in these OMQ languages one
has to pay for polynomial running time in $|\Dmc|$ by an exponential
increase in overall running time. Note however, that the increase is
only in $|q|$, but not in $|\Omc|$. This is in contrast to the case of
AQs where the size of queries is actually constant, and it is good
news as queries tend to be much smaller than ontologies.

Finally we observe that, in $(\ALCF,\text{AQ})$, where \ALCF is the extension
of \ALC with functional roles, OMQ evaluation on databases of
treewidth~2 is \coNP-hard in data complexity when the unique
name assumption is not made.


\smallskip {\bf Related Work.} Monadic second-order logic (MSO) comes
in two versions: MSO$_1$ admits quantification only over sets of nodes
while MSO$_2$ can additionally quantify over sets of
edges. Courcelle's theorem states that MSO$_2$ model checking is FPL
on graphs of bounded treewidth and
\cite{DBLP:journals/mst/CourcelleMR00} shows the same for MSO$_1$ and
graphs of bounded cliquewidth (when a $k$-expression is provided). The
problem studied in the current paper can be translated into MSO$_1$,
but this does not yield the tight complexities presented here. It was
shown in \cite{DBLP:conf/lics/KreutzerT10} that MSO$_2$ is not FPT for
graph classes that have unbounded treewidth and are closed under
substructures (under certain assumptions), see also
\cite{DBLP:journals/jcss/GanianHLORS14}; results in this style do not
appear to be known for MSO$_1$ and cliquewidth. 

\begin{figure}[t]
  \centering
  \begin{tabular}{|l||c|c|c|}
    \hline
    &Bnd. CW & Bnd. TW & Unrestr. \\\hline \hline
    $(\ALC(\Imc),\text{AQ})$ & sngl exp & sngl exp & \ExpTime \\\hline
    $(\text{GF}_2,\text{AQ})$ & dbl exp & sngl exp & \ExpTime \\\hline
    $(\ALC,\text{(U)CQ})$ & dbl exp & open\footnote{} & \ExpTime \\\hline
    $(\ALCI,\text{(U)CQ})$ & dbl exp & dbl exp & \TwoExpTime \\\hline
    $(\text{GF}_2,\text{(U)CQ})$ & open & dbl exp & \TwoExpTime \\\hline
  \end{tabular}
  \caption{Overview of results}
  \label{overview}
\vspace*{-4mm}
\end{figure}
\section{Preliminaries}
\label{sect:prelim}

\footnotetext{This is mistakenly marked as sngl exp in the conference version of this paper.}

{\bf Description Logics.}  Let \NC, \NR, and \Csf be countably
infinite sets of \emph{concept names}, \emph{role names}, and
\emph{constants}. A \emph{role} is a role name $r$ or an \emph{inverse
  role} $r^-$, where $r$ is a role name. Set $(r^-)^- = r$.  
\emph{\ALCI-concepts} are generated by the rule
$$C, D ::= A \mid \neg C \mid C \sqcap D \mid C \sqcup D \mid \exists r.C \mid \forall r.C $$
where $A$ ranges over concept names and $r$ over roles. An
\emph{\ALC-concept} is an \ALCI-concept that does not use
inverse roles.
We use $\top$ to denote $A \sqcup \neg A$ for some fixed concept
name $A$ and $\bot$ for~$\neg \top$.
For $\Lmc \in \{ \ALC, \ALCI \}$, an \emph{\Lmc-ontology} is
a finite set of \emph{concept inclusions (CIs)} $C \sqsubseteq D$ with
$C$ and~$D$ \Lmc-concepts.
A \emph{database} is a finite set of \emph{facts} of the form $A(c)$
or $r(c,c')$ where $A \in \NC \cup \{ \top \}$, $r \in \NR$, and $c,c'
\in \Csf$. We use $\mn{adom}(\Dmc)$ to denote the set of constants
used in a database~\Dmc, also called its \emph{active domain}.
%
The \emph{size} of a syntactic
object $X$, denoted $|X|$, is the number of symbols needed to write
$X$ as a word over a finite alphabet using a suitable encoding. 

The semantics is given in terms of \emph{interpretations}
$\Imc=(\Delta^\Imc,\cdot^\Imc)$, we refer to
\cite{DBLP:books/daglib/0041477} for details. 
%
An interpretation \Imc \emph{satisfies} a CI $C \sqsubseteq D$ if
$C^\Imc \subseteq D^\Imc$, a fact $A(c)$ if $c \in A^\Imc$, and a fact
$r(c,c')$ if $(c,c') \in r^\Imc$. We thus make the standard names
assumption, that is, we interpret constants as themselves. This
implies the unique name assumption. For $S \subseteq \Delta^\Imc$, we
use $\Imc|_S$ to denote the restriction of \Imc to domain~$S$.
An interpretation \Imc is a \emph{model} of an ontology or database if it
satisfies all inclusions or facts in it. A database~$\Dmc$ is
\emph{satisfiable} w.r.t.\ an ontology \Omc if there is a model \Imc
of~\Omc and~\Dmc.

We also consider the \emph{guarded two-variable fragment (GF$_2$)} of
first-order logic. In GF$_2$, only two fixed variables $x$ and $y$ are
available and quantification is restricted to the pattern
$$
\forall \bar{y}(\alpha(\bar{x},\bar{y})\rightarrow \varphi(\bar{x},\bar{y}))
\quad
\exists \bar{y}(\alpha(\bar{x},\bar{y})\wedge \varphi(\bar{x},\bar{y}))
$$
where $\varphi(\bar{x},\bar{y})$ is a GF$_2$ formula with free
variables among $\bar{x} \cup \bar{y}$ and $\alpha(\bar{x},\bar{y})$ is an
atomic formula (possibly an equality) called the \emph{guard} that
uses all variables in $\bar x \cup \bar y$. 
Function symbols and constants are not admitted, and neither is
equality in non-guard positions. We only admit relation symbols of
arity one and two, identifying the former with concept names and the
latter with role names. We may thus interpret GF$_2$ formulas in DL
interpretations. A \emph{GF$_2$-ontology} is a finite set of
GF$_2$-sentences.

\smallskip

{\bf Queries.}
A \emph{conjunctive query (CQ)} is of the form
$q(\bar x) = \exists \bar y\,\varphi(\bar x,\bar y)$, where $\bar x$
and $\bar y$ are tuples of variables and $\varphi(\bar x,\bar y)$ is a
conjunction of \emph{concept atoms} $A(x)$ and \emph{role atoms}
$r(x,y)$, \mbox{$A \in \NC$}, $r \in \NR$, and $x,y$ variables from
$\bar x \cup \bar y$.  We 
call the variables in $\bar x$ the \emph{answer variables} of~$q$, and
use $\mn{var}(q)$ to denote $\bar x \cup \bar y$.  We may write
$\alpha \in q$ to indicate that $\alpha$ is an atom in $q$.
For
$V \subseteq \mn{var}(q)$, we use $q|_V$ to denote the restriction of
$q$ to the atoms that use only variables in~$V$.
A tuple $\bar d \in (\Delta^\Imc)^{|\bar x|}$
is an \emph{answer} to $q$ on an interpretation \Imc 
if there is a homomorphism $h$ from $q$ to
\Imc with $h(\bar x)=\bar d$. More details are in the appendix.

%
 

A {\em union of conjunctive queries (UCQ)} $q(\bar x)$ is a disjunction
of CQs with the same answer variables $\bar x$. A tuple
$\bar d \in (\Delta^\Imc)^{|\bar x|}$ is an \emph{answer} to $q$ on
interpretation \Imc, written $\Imc \models q(\bar d)$,
if $\bar d$ is an answer to
some CQ in $q$ on \Imc. 
We use $q(\Imc)$ to denote set of all
answers to $q$ on \Imc.  The \emph{arity} of $q$ is the length of
$\bar x$ and $q$ is \emph{Boolean} if it is of arity zero.
An \emph{atomic query (AQ)} is a CQ of the form $A(x)$ with $A$ a concept
name.

\smallskip
{\bf Ontology-Mediated Querying.}
An \emph{ontology-mediated query (OMQ)} is a pair
$Q=(\Omc,q)$ with \Omc an ontology 
and $q$ a query such as a UCQ. While OMQs are often defined to
include
an additional third component, the data signature, the problem of
query evaluation studied in this paper is insensitive to that
component,
and so we omit it. We write $Q(\bar
x)$ to indicate that the answer variables of $q$ are $\bar x$.
%
A tuple $\bar a \in \mn{adom}(\Dmc)^{|\bar x|}$ is an \emph{answer} to
$Q(\bar x)$ on a database \Dmc, written $\Dmc \models Q(\bar a)$, if
$\Imc \models q(\bar a)$ for all models \Imc of \Omc and~\Dmc.  When
more convenient, we might alternatively write
$\Dmc,\Omc \models q(\bar a)$. We write $Q(\Dmc)$ to denote
the set of all answers to $Q$ on \Dmc. 
We use $(\Lmc,\Qmc)$ to denote the \emph{OMQ language} that
contains all OMQs~$Q$ in which $\Omc$ is formulated in the DL \Lmc and $q$
in the query language \Qmc, such as in $(\ALCI,\text{UCQ})$ and $(\text{GF}_2,\text{AQ})$.

Let $(\Lmc,\Qmc)$ be an OMQ language.  We write \emph{OMQ evaluation
  in} $(\Lmc,\Qmc)$ to denote the problem to decide, given an OMQ
$Q(\bar x) \in (\Lmc,\Qmc)$, a database \Dmc, and a tuple
$\bar c \in \mn{adom}(\Dmc)^{|\bar x|}$, whether $\bar c \in
Q(\Dmc)$. A \emph{parameterization} is a function $\kappa$ that
assigns to every input $Q,\Dmc,\bar c$ a parameter
$\kappa(Q,\Dmc,\bar c) \in \mathbb{N}$. We will use $\kappa=|Q|$ in
lower bounds and $\kappa=|Q|+k$ in upper bounds where $k$ is the
cliquewidth of \Dmc as defined below.  OMQ evaluation is
\emph{fixed-parameter tractable (FPT) for} $\kappa$ if it can be
decided by an algorithm that 
runs in time $f(\kappa(Q,\Dmc,\bar c)) \cdot \mn{poly}(n)$
where $f$ is a computable function, here and throughout the paper, 
$\mn{poly}$ denotes an unspecified 
polynomial, and $n$ is the size of the input. \emph{Fixed-parameter linearity (FPL) for} $\kappa$ is 
defined analogously, with running time 
$f(\kappa(Q,\Dmc,\bar c)) \cdot O(n)$. 

\smallskip {\bf Treewidth and Cliquewidth.} Treewidth describes how
similar a graph (in our case: a database) is to a tree, and
cliquewidth does the same for cliques
\cite{DBLP:books/daglib/0030804}. 
%
Due to space restrictions and since this article focuses on
cliquewidth, we provide the definition of treewidth only in the
appendix. 
Databases of cliquewidth $k$ can be constructed by a $k$-expression as
defined below, where $k$ refers to the number of \emph{labels} in the
expression. For convenience, we use a reservoir of fresh concept names
$L_1,L_2,\dots$ as labels that may occur in databases, but not 
in ontologies and queries.
%
We say that
a constant $c \in \mn{adom}(\Dmc)$ is \emph{labeled} $i$ if $L_i(c) \in \Dmc$.
%
An \emph{expression} $s$ defines a 
database $\Dmc_s$ built from four operations, as follows:

\smallskip
\noindent {$\boldsymbol{s = \iota(\Dmc)}$} is the nullary operator for
\emph{constant introduction} where $i \geq 1$ and \Dmc is a 
database with $|\mn{adom}(\Dmc)|=1$ that  
contains
exactly one fact of the form $L_i(c)$, and $\Dmc_s=\Dmc$.

\smallskip
\noindent
{$\boldsymbol{s = s_1 \oplus s_2}$} is the binary operator for
\emph{disjoint union} where $s_1$ and $s_2$ are expressions such that
$\mn{adom}(\Dmc_{s_1}) \cap \mn{adom}(\Dmc_{s_2}) = \emptyset$
and $\Dmc_s = \Dmc_{s_1} \cup \Dmc_{s_2}$.

\smallskip
\noindent {$\boldsymbol{s = \alpha_{i, j}^r(s')}$} is the unary
operator for \emph{role insertion} where $r$ is a role name,
$i,j \geq 1$ are distinct, and $s'$ is an expression.  It links all constants
labeled $i$ to all constants labeled $j$ using role~$r$, that
is, 
$$\Dmc_s = \Dmc_{s'} \cup \{r(a,b) \mid L_i(a) \in \Dmc_{s'} \text{ and } L_j(b) \in \Dmc_{s'}\}.$$

\smallskip
\noindent {$\boldsymbol{s = \rho_{i \rightarrow j}(s')}$} is the unary
operator for \emph{relabeling} where $i,j \geq 1$ are distinct and
$s'$ is an expression. It changes all $i$-labels to $j$, that is,
$$
\Dmc_s = \Dmc_{s'} \setminus \{L_i(c) \mid c \in \mn{adom}(\Dmc_{s'})\}
\cup \{L_j(c) \mid L_i(c) \in \Dmc_{s'}\}.
$$

\smallskip 
An expression $s$ is a \emph{$k$-expression} if all labels in $s$ are from
$\{ L_1,\dots,L_k\}$. A database \Dmc has \emph{cliquewidth
  $k>0$} if there is a $k$-expression $s$ with $\Dmc_s=\Dmc$, but no
$k'$-expression $s'$ such that $k' < k$ and $\Dmc_{s'}=\Dmc$, up to
facts of the form $L_i(c)$.

\begin{example}
  \label{ex:teacher_pupils_schools}
  Consider the class of databases about schools that use the concept names
  $\mn{Pupil}$, $\mn{Teacher}$ and $\mn{School}$ as well as role names
  $\mn{teaches}$, $\mn{worksAt}$ and $\mn{isClassmateOf}$, and where
  each teacher works at exactly one school and can teach any number of
  pupils.
  The set of all pupils is partitioned into groups of pupils that
  are classmates of one another, and each teacher teaches exactly one
  such group.  A basic database from that class is
  \begin{align*}
    \Dmc =\ &\{\mn{Pupil}(a_1), \mn{Pupil}(a_2), \mn{Teacher}(b), \mn{School}(c),\\[0.5mm]
    &\ \ \mn{worksAt}(b, c), \mn{teaches}(b, a_1), \mn{teaches}(b, a_2),\\[0.5mm]
    &\ \ \mn{isClassmateOf}(a_1, a_2), \mn{isClassmateOf}(a_2, a_1)\}.
  \end{align*}
  This database has cliquewidth $3$, and in fact so does any database
  from the described class, independently of the number of pupils,
  teachers, and schools.  When constructing $k$-expressions, it
  suffices to use up to two labels for pupils, one for teachers, and
  one for schools. With proper relabeling, no more than three labels
  are needed at the same time.
\end{example}

\section{Upper Bounds and Algorithms}
\label{sect:upper}

We prove that OMQ evaluation on databases of bounded cliquewidth is
FPL when the parameter is the size of the OMQ plus the cliquewidth of
the database and establish upper  bounds on the overall running time. 

\subsection{\ALCI with AQs}
\label{sect:alciupper}

%
\begin{restatable}{theorem}{alciaqubthm}
  \label{thm:alc1exp}
  In $(\ALCI, \text{AQ})$, an OMQ
  $Q=(\Omc,q)$ can be evaluated on a database \Dmc of cliquewidth
  $k$ in time~\mbox{$2^{O(|\Omc| \cdot k)} \cdot
    |\Dmc|$.} 
\end{restatable}
Theorem~\ref{thm:alc1exp} implies that OMQ evaluation in
$(\ALCI, \text{AQ})$ is in FPL, and also that it is in linear time in
data complexity when the cliquewidth of databases is bounded by a
constant.
As announced before, we assume here and in all subsequent upper
complexity bounds that a $k$-expression $s_0$ for \Dmc is given as
part of the input. In fact, the number of subexpressions of $s_0$ must
be $O(|\Dmc|)$, a condition that is satisfied by any reasonable
$k$-expression for \Dmc. These assumptions could clearly be dropped
if, given a database \Dmc of cliquewidth~$k$, we could compute a
$k$-expression that generates \Dmc (or a sufficient approximation
thereof) in time $2^{O(k)} \cdot |\Dmc|$. It is an open problem
whether this is possible, see
e.g.~\cite{downey2013fundamentals}. There is, however,
an 
algorithm that computes, given a database \Dmc of cliquewidth~$k$, in
time $|\mn{adom}(\Dmc)|^3$ a $k'$-expression with $k' \leq 2^{3k-1}$
\cite{oum2008approximating}. As a consequence, we obtain an analogue of
Theorem~\ref{thm:alc1exp} that does not require a $k$-expression
for \Dmc to be given and where the time bound is replaced with
$2^{|Q| \cdot 2^{O(k)}} \cdot |\Dmc|^3$. While this no longer yields
FPL, it still yields FPT.

To prove Theorem~\ref{thm:alc1exp}, it suffices to give an algorithm
for database satisfiability w.r.t.\ an \ALCI ontology that runs within the
stated time bounds. To decide whether
\mbox{$\Dmc,\Omc \models A(c)$}, we may then simply check whether
$\Dmc \cup \{\overline{A}(c)\}$ is unsatisfiable w.r.t.\ the ontology
$\Omc \cup
\{\overline{A} \equiv \neg A\}$.

Assume that we are given as input a database $\Dmc_0$, an ontology
\Omc, and a $k$-expression $s_0$ that generates $\Dmc_0$. We may
assume w.l.o.g.\ that \Omc takes the form
$\{ \top \sqsubseteq C_\Omc \}$ where $C_\Omc$ is an \ALCI-concept in
negation normal form (NNF), that is, negation is only applied to
concept names, but not to compound concepts. Every \ALCI-ontology \Omc
can be converted into this form in linear time
\cite{DBLP:books/daglib/0041477}. 

Our algorithm traverses the $k$-expression
$s_0$ bottom-up, computing for each  subexpression~$s$ a succinct
representation of the models of the database $\Dmc_s$ and ontology
\Omc. Before giving details, we introduce some relevant notions.
We use $\mn{sub}(\Omc)$ to denote the set of all concepts in \Omc,
closed under subconcepts and
%
$\mn{cl}(\Omc)$ for the
extension of $\mn{sub}(\Omc)$ with the NNF of all concepts $\neg C$,
$C \in \mn{sub}(\Omc)$. Moreover, $\mn{cl}^\forall(\Omc)$ denotes the
restriction of $\mn{cl}(\Omc)$ to concepts of the form $\forall r .C$
and
$\mn{cl}^\ast(\Omc)$ is  $\{ C \mid \forall r . C \in \mn{cl}(\Omc) \}$.
A \emph{type for \Omc} is a set $t \subseteq \mn{cl}(\Omc)$ such that
for some model \Imc of \Omc and some $d \in \Delta^\Imc$,
$
  t = \{ C \in \mn{cl}(\Omc) \mid d \in C^\Imc\}.
$ We then also denote $t$ with $\mn{tp}_\Imc(d)$.

%
An \emph{input output assignment (IOA)} for \Omc is a pair
$\gamma = (\mn{in},\mn{out})$ with
$\mn{in}: \{1,\dots,k\} \rightarrow 2^{\mn{cl}^\ast(\Omc)}$ and
$\mn{out}: \{1,\dots,k\} \rightarrow
2^{\mn{cl}^\forall(\Omc)}$ total functions. 
For easier reference, we use
$\gamma^{\mn{in}}$ to denote $\mn{in}$ and $\gamma^{\mn{out}}$ to denote \mn{out}.
Every database \Dmc and model \Imc of \Dmc and \Omc give rise to an IOA
$\gamma_{\Imc,\Dmc}$ for \Omc defined by setting, for $1 \leq i \leq k$,
\begin{align*}
  \gamma^{\mn{in}}_{\Imc,\Dmc}(i) &= \mn{cl}^\ast(\Omc) \cap \bigcap_{L_i(c) \in \Dmc}
                     \mn{tp}_\Imc(c) \text{ and} \\
  \gamma^{\mn{out}}_{\Imc,\Dmc}(i) &= \mn{cl}^\forall(\Omc) \cap \bigcup_{L_i(c) \in \Dmc} \mn{tp}_\Imc(c).
\end{align*}
%
Intuitively, $\gamma^{\mn{out}}(i)$ lists outputs generated by label
$i$ in the sense that every concept
$\forall r . C \in \gamma^{\mn{out}}(i)$ `outputs' $C$ to all
constants labeled $j$ when we use the $\alpha^r_{i,j}$ operation to
introduce new role edges. It
is then important that $C \in \gamma^{\mn{in}}(j)$, and in this sense
$\gamma^{\mn{in}}(j)$ lists inputs accepted by label $j$.

The central idea of our algorithm is to compute, for each
subexpression $s$ of $s_0$, the set of IOAs
$$\Theta(s)=\{ \gamma_{\Imc,\Dmc_s} \mid \Imc \text{ model of } \Dmc_s \text{
  and } \Omc \}.
$$
It then remains to check whether $\Theta(s_0)$ is
non-empty. The sets $\Theta(s)$ are computed as follows:

\begin{enumerate}[align=left]
\item[$\boldsymbol{s = \iota(\Dmc)\text{: }}$] Set $\Theta(s)=
  \{ \gamma_{\Imc,\Dmc} \mid \Imc \text{ model of } \Dmc \text{ and }
  \Omc
  \}$.
\end{enumerate}


\begin{enumerate}[align=left]

\item[$\boldsymbol{s = s_1 \oplus s_2\text{: }}$] $\Theta(s)$ contains
  an IOA $\gamma$ for each pair of IOAs
  $(\gamma_1, \gamma_2) \in \Theta(s_1) \times \Theta(s_2)$ defined by
  setting, for $1 \leq i \leq k$:
  $$
    \gamma^{\mn{in}}(i) = \gamma_1^{\mn{in}}(i)
    \cap \gamma_2^{\mn{in}}(i) \quad
    \gamma^{\mn{out}}(i) = \gamma_1^{\mn{out}}(i) \cup
    \gamma_2^{\mn{out}}(i).
  $$

\item[$\boldsymbol{s = \alpha_{i, j}^r(s')\text{: }}$] $\Theta(s)$ is
  obtained from $\Theta(s')$ by removing IOAs that are ruled out by
  the additional role links. Formally, we keep in $\Theta(s)$ only
  those
  IOAs $\gamma$ such that 
  \begin{enumerate}[label=\emph{\alph*})]
        \item if $\forall r.C \in \gamma^{\mn{out}}(i)$ then $C \in \gamma^{\mn{in}}(j)$ and
        \item if $\forall r^-.C \in \gamma^{\mn{out}}(j)$ then $C \in \gamma^{\mn{in}}(i)$.
        \end{enumerate}

      \item[$\boldsymbol{s = \rho_{i \rightarrow j}(s')\text{: }}$]
        for each
${\widehat\gamma} \in \Theta(s')$, 
        $\Theta(s)$ contains the  IOA $\gamma$ defined as follows:
        \begin{enumerate}[label=\emph{\alph*}),start=3]
        \item $\gamma^{\mn{in}}(i) = \mn{cl}^\ast(\Omc)$ and
          $\gamma^{\mn{out}}(i) = \emptyset$
        \item $\gamma^{\mn{in}}(j) = {\widehat\gamma}^{\mn{in}}(i) \cap
          {\widehat\gamma}^{\mn{in}}(j)$ and \\
          $\gamma^{\mn{out}}(j) = {\widehat\gamma}^{\mn{out}}(i) \cup
          {\widehat\gamma}^{\mn{out}}(j)$
        \item $\gamma^{\mn{in}}(\ell) =
          {\widehat\gamma}^{\mn{in}}(\ell)$
and $\gamma^{\mn{out}}(\ell) =
          {\widehat\gamma}^{\mn{out}}(\ell)$  \\
          $\text{for all } \ell \in \{1,\dots,k\} \setminus \{i, j\}$.
        \end{enumerate}
\end{enumerate}
%
%
%
We prove in the appendix that the algorithm is correct.
\begin{restatable}{lemma}{alciaqublem}
  \label{lem:1exp-correctness}
  For all subexpressions $s$ of $s_0$,
  $\Theta(s)=\{ \gamma_{\Imc,\Dmc_s} \mid \Imc \text{ model of }
  \Dmc_s \text{ and } \Omc \}$.
\end{restatable}
The presented algorithm achieves the intended running time. We do a
single bottom-up pass over $s_0$, considering each subexpression $s$,
of which there are at most $O(|\Dmc|)$ many. The size of each IOA in
$\Theta(s)$ is $O(|\Omc| \cdot k)$ and there are at most
$2^{O(|\Omc| \cdot k)}$ IOAs. It can be verified that each set
$\Theta(s)$ can be constructed in time $2^{O(|\Omc| \cdot k)}$. For
the case $s = \iota(\Dmc)$, this essentially amounts to enumerating all
types for \Omc, more details are in the appendix.

\subsection{\ALCI with UCQs}
\label{sect:alciucq}

\begin{restatable}{theorem}{alciucqtwoexpthm}
  \label{thm:alciucq2exp}
  In $(\ALCI, \text{UCQ})$, an OMQ $Q =(\Omc,q)$
  can be evaluated
  on a database \Dmc of cliquewidth $k$ in time
  $2^{|\Omc| \cdot 
    k^{O(|q| \log(|q|))}} \cdot |\Dmc|$.
\end{restatable}
The algorithm used to prove Theorem~\ref{thm:alciucq2exp} is a
generalization of the one used for
Theorem~\ref{thm:alc1exp}. It also traverses the
  $k$-expression $s_0$ bottom-up, computing for each subexpression $s$
  a set of IOAs that additionally are annotated with information about
  partial homomorphisms from CQs in the UCQ $q$ into models of
  $\Dmc_s$ and \Omc. The generalization is not entirely
  straightforward. First, the homomorphism may map some variables to
  elements of the model that are `outside' of the database $\Dmc_s$
  while IOAs only store information about constants from $\Dmc_s$. And
  second, we cannot even memorize all homomorphisms that map all
  variables into $\Dmc_s$ as there may be $|\Dmc_s|^{|q|}$ many of
  them, too many for FPL.  Our solution to the first issue is to
  modify the ontology and UCQ so that variables need never be mapped
  outside of \Dmc. We resolve the second issue by the observation that
  it suffices to memorize the \emph{labels} of the constants in
  $\Dmc_s$ to which variables are mapped, but not their precise
  identity.  The generalized algorithm establishes
  Theorem~\ref{thm:alc1exp} as a special case; we singled
  out the algorithm in the previous section for didactic reasons.

Assume that we are given as input an OMQ
$Q(\bar x)=(\Omc,q) \in (\ALCI, \text{UCQ})$, a database $\Dmc_0$, and
a $k$-expression $s_0$ that generates $\Dmc_0$. It is easy to see that
we may assume w.l.o.g.\ that $Q$ is Boolean, details are in the
appendix.
We may also assume the ontology~\Omc to
contain a CI $\top \sqsubseteq A_\top$, with $A_\top$ not used
anywhere else in~\Omc, and that every CQ $p$ in $q$ contains the
concept atom $A_\top(x)$ for every $x \in \mn{var}(p)$. This will
prove useful for the constructions below.


An interpretation \Imc is a \emph{tree} if the undirected graph
$G_\Imc = (V,E)$ with $V=\Delta^\Imc$ and  
$E = \{ \{x,y \} \mid (x,y) \in r^\Imc \text{ for some } r \}$ 
is a tree and there are no self loops and multi-edges, the latter
meaning that $(d,e) \in r_1^\Imc$ implies $(d,e) \notin r_2^\Imc$ for
all distinct roles $r_1,r_2$.  We call a model \Imc of $\Dmc_0$
\emph{tree-extended} if it satisfies the following conditions:
\begin{itemize}

\item $r^\Imc \cap (\mn{adom}(\Dmc_0) \times \mn{adom}(\Dmc_0)) = \{(d, e) \mid r(d, e) \in \Dmc_0\}$
for all role names $r$ and

\item if \Imc is modified by setting
  $r^\Imc = r^\Imc \setminus (\mn{adom}(\Dmc_0) \times
  \mn{adom}(\Dmc_0))$ for all role names $r$, then the result is a
  disjoint union of trees and each of these trees contains exactly
  one constant from $\mn{adom}(\Dmc_0)$. 

\end{itemize}
The following is well-known, see for example \cite{DBLP:conf/dlog/Lutz08}.
\begin{lemma}
  \label{lem:treeextended1}
  If there is a model \Imc of $\Dmc_0$ and \Omc such that
  $\Imc \not \models q$, then there is such a model \Imc that
  is tree-extended.
\end{lemma}
Lemma~\ref{lem:treeextended1} provides the basis for the announced
modification of \Omc and the CQs in $q$: we extend \Omc to an ontology
$\Omc_q$ and $q$ into a UCQ $\widehat q$ such that for every
tree-extended model $\Imc$ of $\Dmc_0$ and $\Omc_q$, $\Imc \models q$
iff $\Imc|_{\mn{adom}(\Dmc_0)} \models \widehat q$. We then proceed to
work with $\Omc_q$ and $\widehat q$ in place of \Omc and $q$, which
allows us to concentrate on homomorphisms from CQs in $\widehat q$ to
models \Imc of \Dmc and $\Omc_q$ that map all variables to
$\mn{adom}(\Dmc_0)$.

We next present the extension of \Omc.  A Boolean or unary CQ $p$
being a \emph{tree} (without self-loops and multi-edges) is defined in
the expected way. Note that a unary tree CQ $p$ can be viewed as an
\ALCI-concept~$C_p$ in an obvious way. For example, if
\mbox{$p(x) = \{A(x),r(y,x),s(y,z),B(z),r(y,z'),A(z')\}$}, then
$C_p=A\sqcap \exists r^-.(\exists s.B\sqcap \exists r.A)$.  
We use $\mn{trees}(q)$ to denote the set of all
(Boolean) tree CQs that can be obtained from a CQ in $q$ by first
dropping atoms and then taking a contraction. When writing $p(x)$ with
$p \in \mn{trees}(q)$ and $x \in \mn{var}(p)$, we mean $p$ with $x$
made answer variable.
%
%
The ontology $\Omc_q$ is obtained
from \Omc by adding the following, and then converting to NNF:
\begin{itemize}

\item for every $p \in \mn{trees}(q)$ and $x \in \mn{var}(p)$, the CIs
  \mbox{$A_{p(x)} \sqsubseteq C_{p(x)}$} and $C_{p(x)} \sqsubseteq A_{p(x)}$
where $A_{p(x)}$ is a fresh concept name;

  \item for every $p \in \mn{trees}(q)$, the CIs
    $A_{p} \sqsubseteq \exists r_{p} . C_{p}$,
    $C_{p} \sqsubseteq A_p$, and $\exists r . A_p \sqsubseteq A_p$
    for every role $r$ used in~\Omc, where $A_p$ is a fresh concept
    name and $r_p$ a fresh role name.


\end{itemize}
Note that $|\Omc_q| \in |\Omc| \cdot 2^{O(|q| \cdot \log
  |q|)}$.


Towards defining the UCQ $\widehat q$, consider all CQs that can be
obtained as follows. Start with a contraction $p$ of a CQ from~$q$,
then choose a set of variables $S \subseteq \mn{var}(p)$ such
that 
$$p^-=p \setminus \{ r(x,y) \in p \mid x,y \in S \}$$ is a disjoint union of
tree CQs each of which contains at most one variable from
$S$ and at least one variable that is not from~$S$.  Include in
$\widehat q$ all CQs that can be obtained by extending $p|_S$ as
follows:
\begin{enumerate}

\item for every maximal connected component $p'$ of $p^-$ that
  contains a (unique) variable $x_0 \in S$, 
  add the atom $A_{p'(x_0)}(x_0)$;


\item for every maximal connected component $p'$ of $p^-$ that
  contains no variable from $S$, add the atom $A_{p'}(z)$ with $z$
  a fresh variable.
  

\end{enumerate}
%
%
It is easy to verify that in Points~1 and~2 above, the CQ~$p'$ is in $\mn{trees}(q)$ and thus CIs for $p'$ have been introduced
in the construction of~$\Omc_q$.
The number of CQs in $\widehat q$ is single exponential in $|q|$ and
each CQ is of size at most~$|q|$. 
%
\begin{restatable}{lemma}{qhatworksalclem}
  \label{lem:qhatworksalc}
  For every tree-extended model \Imc of $\Dmc_0$ and $\Omc_q$,
  $\Imc \models q$ iff $\Imc|_{\mn{adom}(\Dmc_0)} \models \widehat q$.
\end{restatable}
Let \Qmc denote the set of all subqueries $p$ of CQs in $\widehat q$,
that is, $p$ can be obtained from a CQ in $\widehat q$ by dropping
atoms.
%
%
%
A \emph{decorated IOA} for $\Omc_q$ is a triple
$\gamma=(\mn{in},\mn{out},S)$ where $(\mn{in},\mn{out})$ is an IOA
and 
$S$ is a set of pairs $(p,f)$ with $p \in \Qmc$ and
$f:\mn{var}(p) \rightarrow \{1,\dots,k\}$ a total function.  We use
$\gamma^{\mn{in}}$ to denote $\mn{in}$ and likewise for
$\gamma^{\mn{out}}$ and $\gamma^S$. We argue in the appendix that the
number of decorated IOAs is
$2^{|\Omc| \cdot 
  k^{O(|q| \log
    |q|)}}$.

Every database \Dmc and model \Imc of \Dmc and $\Omc_q$ give rise to a
decorated IOA $\gamma _{\Imc,\Dmc}$ for $\Omc_q$ where
\begin{itemize}

\item $\gamma^{\mn{in}}_{\Imc,\Dmc}$ and $\gamma^{\mn{out}}_{\Imc,\Dmc}$ are
  defined as in the previous section and

\item $\gamma^S$ contains all pairs $(p,f)$ with $p \in \Qmc$ and $f$ a
  function from $\mn{var}(p)$ to $\{1,\dots,k\}$ such that there is a
  homomorphism $h$ from $p$ to \Imc with $f(x)=i$ iff
  $L_i(h(x)) \in \Dmc$ for all variables $x \in \mn{var}(p)$.
  
\end{itemize}
For all
subexpressions $s$ of $s_0$, our algorithm computes
$$\Theta(s)=\{ \gamma_{\Imc,\Dmc_s} \mid \Imc \text{ is a tree-extended model of } \Dmc_s \text{
  and } \Omc_q \}.$$ By Lemma~\ref{lem:qhatworksalc} and since every
tree-extended model of $\Dmc_0$ and \Omc can be extended to a
tree-extended model of $\Dmc_0$ and~$\Omc_q$, this implies that, as
desired, $\Dmc_0,\Omc \not\models q$ iff there is a decorated IOA
$\gamma \in \Theta(s_0)$ such that $\gamma^S$ contains no pair $(p,f)$
with $p$ a CQ in $\widehat q$.
%
The sets $\Theta(s)$ are computed as follows:
\begin{enumerate}[align=left]
\item[$\boldsymbol{s = \iota(\Dmc)\text{: }}$] $\Theta(s)=
  \{ \gamma_{\Imc,\Dmc} \mid \Imc$ is a tree-extended model of $\Dmc$ and  $\Omc_q
  \}$. 
\end{enumerate}


\begin{enumerate}[align=left]

\item[$\boldsymbol{s = s_1 \oplus s_2\text{: }}$] $\Theta(s)$ contains
  a decorated IOA $\gamma$ for each pair
  $(\gamma_1, \gamma_2) \in \Theta(s_1) \times \Theta(s_2)$ where, for
  $1 \leq i \leq k$:
%
  $$
    \gamma^{\mn{in}}(i) = \gamma_1^{\mn{in}}(i)
    \cap \gamma_2^{\mn{in}}(i) \quad
    \gamma^{\mn{in}}(i) = \gamma_1^{\mn{out}}(i) \cup
    \gamma_2^{\mn{out}}(i)).
  $$
    and
    $\gamma^S=\gamma_1^S \cup \gamma_2^S \cup \{ (p_1 \cup p_2,f_1 \cup f_2) 
    \mid (p_i,f_i) \in \gamma^S_i$ for $i \in \{1,2\}$, $p_1 \cup 
    p_2 \in \Qmc$ and $\mn{var}(p_1) \cap \mn{var}(p_2) = \emptyset \}$. 

  \item[$\boldsymbol{s = \alpha_{i, j}^r(s')\text{: }}$] $\Theta(s)$
    contains the decorated IOAs $\gamma$ that can be obtained by
    choosing a
    $\widehat\gamma \in \Theta(s')$
    that satisfies Conditions~a) and~b) from
    Section~\ref{sect:alciupper} and then defining
     \begin{itemize}

     \item  $\gamma^{\mn{in}}={\widehat\gamma}^{\mn{in}}$ and
       $\gamma^{\mn{out}}={\widehat\gamma}^{\mn{out}}$;

     \item $\gamma^S$ to be ${\widehat \gamma}^S$ extended with all pairs
       $(p',f)$ such that \mbox{$p' \in \Qmc$} and for some
       $(p,f) \in \widehat S$, $p' $ can be obtained from $p$ by
       adding zero or more atoms $r(x,y)$, subject to the condition
       that $f(x)=i$ and $f(y)=j$.

     \end{itemize}

   \item[$\boldsymbol{s = \rho_{i \rightarrow j}(s')\text{: }}$] for
     each
     $\widehat\gamma \in \Theta(s')$,
     $\Theta(s)$ contains the decorated IOA $\gamma$
     obtained by
     defining
        \begin{itemize}
        \item $\gamma^{\mn{in}}$ and $\gamma^{\mn{out}}$
          according to Conditions~c) to~e) in
          Section~\ref{sect:alciupper};

        \item $\gamma^S$ to contain all pairs $(p,f')$ such that for
          some \mbox{$(p,f) \in \widehat S$}, $f'$ can be obtained by
          starting with $f$ and then setting $f'(x)=j$ whenever
          $f(x)=i$.          
        \end{itemize}
\end{enumerate}
Note that \Qmc contains $A_\top(x)$ for every variable 
$x$ in $q$ and thus we can start with mapping subqueries
that contain only a single variable $x$ (and a dummy atom)
in the $s = \iota(\Dmc)$ case.
The algorithm achieves the intended goal.
\begin{restatable}{lemma}{alciucqcorrectnesslem}
  \label{lem:alci-ucq-correctness}
  For all subexpressions $s$ of $s_0$,
  $\Theta(s)=\{ \gamma_{\Imc,\Dmc_s} \mid \Imc \text{ is a tree-extended model of }
  \Dmc_s \text{ and } \Omc_q \}$.
\end{restatable}
The running time is analyzed in the appendix.

\subsection{\GFtwo with AQs}

%
%
\begin{restatable}{theorem}{gftwoaqtwoexpthm}
  \label{thm:gf2aq2exp}
  In $(\text{GF}_2, \text{AQ})$, an OMQ $Q =(\Omc,q)$ can be evaluated
  on a database \Dmc of cliquewidth $k$ in time
  $2^{2^{{O(|\Omc|)}}\cdot k^2} \cdot |\Dmc|$. 
\end{restatable}
As in Section~\ref{sect:alciupper}, we may concentrate on the
satisfiability
of databases w.r.t.\ ontologies. To decide whether
$\Dmc,\Omc \models A(\bar c)$, we then simply check whether
$\Dmc \cup \{\overline{A}(\bar c)\}$ is unsatisfiable w.r.t.\ the ontology
$\Omc \cup
\{
\forall \bar x \, (\overline{A}(\bar x) \rightarrow \neg {A}(\bar x))
\}$. 
Note that this also captures `binary AQs' $A(\bar
c)$. 

Assume that we are given as input a
GF$_2$-ontology \Omc, a database $\Dmc_0$, and a $k$-expression
$s_0$ that generates $\Dmc_0$.  For $i \in
\{1,2\}$, we use
$\mn{cl}_i(\Omc)$ to denote the set of all subformulas of \Omc with
at most $i$ free variables, closed under single negation. For
$i=2$, we additionally assume that
$\mn{cl}_i(\Omc)$ contains all formulas $r(x,y)$ and $r(y,x)$ with
$r$ a role name that occurs in
$\Dmc_0$.\footnote{Recall that $x$ and $y$ are the two fixed
  (and distinct) variables admitted in GF$_2$.}
An
\emph{$i$-type for \Omc} is a set $t \subseteq \mn{cl}_i(\Omc)$ such
that for some model \Imc of \Omc and some
$\bar d \in (\Delta^\Imc)^i$,
\begin{align*}
  t &=  \{ \varphi(\bar x) \in \mn{cl}_i(\Omc) \mid \Imc
                             \models \vp(\bar e), \\
  & \qquad\qquad \bar e \text{ a subtuple of } \bar d \text{ of length
    } |\bar x| \in \{0,\dots,i\}\}.
\end{align*}
%
We then also denote $t$ with $\mn{tp}^i_\Imc(\bar d)$ and
use $\mn{TP}_i$
to denote the set of all $i$-types for \Omc. For
$t_1,t_2 \in \mn{TP}_1$ and $t \in \mn{TP}_2$, we say that $t_1$ and
$t_2$ \emph{are compatible with}~$t$ and write
$t_1 \rightsquigarrow_{t} t_2$ if there is a model \Imc of \Omc and
$d_1,d_2 \in \Delta^\Imc$ such that $\mn{tp}_\Imc^1(d_i)=t_i$ for $i \in
\{1,2\}$ and \mbox{$\mn{tp}_\Imc^2(d_1,d_2)=t$}. Note that every
1-type contains a 0-type and every 2-type contans a 0-type and two
1-types, identified by the subformulas that use free variable $x$ and
$y$, respectively.

We again traverse the $k$-expression $s_0$ bottom-up, computing for
each subexpression~$s$ a representation of the models of the database
$\Dmc_s$ and ontology \Omc. This representation, however,
  is different from the ones in the previous sections and does not use
  IOAs. The reason is as follows. A concrete model \Imc of $\Dmc_s$
  and \Omc assigns a 1-type to every constant in $\Dmc_s$ and thus a
  set $T_i$ of 1-types to every label $i$. In \ALCI, we may represent
  each set $T_i$ by the sets
  $\{ \forall r . C \mid \exists t \in T_i: \forall r . C \in t\}$ and
  $\{ C \in \mn{cl}^*(\Omc) \mid \forall t \in T_i: C \in t\}$ of an
  IOA because they contain all information that we need when using the
  $\alpha^r_{i,j}$ operator to introduce new edges. The virtue is that
  there are only single exponentially many IOAs. Now consider the
  GF$_2$-ontology \Omc that contains the sentence
  $$
      \forall x \forall y \, r(x,y) \rightarrow \bigwedge_{1 \leq i
        \leq n} (A_i(x) \leftrightarrow A_i(y)).
  $$
  The 1-types in $T_i$ may contain any combination of the formulas
  $A_1(x),\dots,A_n(x)$ and for dealing with the $\alpha^r_{i,j}$
  operator we clearly have to know all of them. However, there are
  double exponentially many sets of such combinations, and this
  results in a double exponential upper bound.  

We now give
details of how we represent models.  A \emph{multi-edge} is a set of
atoms $r(x,y)$ and $r(y,x)$. Let \Mmc denote the set of all
multi-edges that only use role names that occur in $\Dmc_0$.  For a
database \Dmc and distinct $a,b \in \mn{adom}(\Dmc)$, we use
$M_\Dmc(a,b)$ to denote the multi-edge that consists of all atoms
$r(x,y)$ such that \Dmc contains the fact $r(a,b)$, and $r(y,x)$ such
that $\Dmc$ contains the fact $r(b,a)$.

A \emph{model abstraction} is a pair $\gamma=(T,E)$ where
\begin{itemize}


\item $T: \{1,\dots,k\} \rightarrow 2^{\mn{TP}_1}$ is a partial
  function that associates each label with a set of 1-types;

\item
  $E: \{1,\dots,k\}^2 \rightarrow 2^{\mn{TP}_1 \times \Mmc \times
    \mn{TP}_1}$ is a partial function that associates each pair of labels
  $(i,j)$, where $i \neq j$, with a set of triples $(t_1,M,t_2)$ from
  $\mn{TP}_1 \times \Mmc \times \mn{TP}_1$
  
\end{itemize}
such that all 1-types in the ranges of $T$ and $E$ contain exactly the
same subformulas of \Omc that are sentences, that is, they contain the
same 0-type. We use $\gamma^0$ to denote this 0-type, $\gamma^T$ to
denote $T$, and likewise for~$\gamma^E$. We use
  multi-edges in model abstractions rather than 2-types as the former
  allow us to add edges later when the operator $\alpha^r_{i,j}$ is
  applied whereas the latter fix all edges from the beginning.

Every database \Dmc and model \Imc of
\Dmc give rise to a model abstraction $\gamma_{\Imc,\Dmc}$ defined by
setting, for \mbox{$1 \leq i,j \leq k$},
\begin{align*}
  \gamma^T_{\Imc,\Dmc}(i) &= \{ \mn{tp}^1_\Imc(c) \mid L_i(c) \in \Dmc \}
                    \\
  \gamma^E_{\Imc,\Dmc}(i,j) &= \{ (\mn{tp}^1_\Imc(c),
                              M_\Dmc(c,c'),\mn{tp}^1_\Imc(c')) \mid L_i(c) \in \Dmc\\
  & \qquad\qquad\qquad  \text{ and } L_j(c') \in \Dmc
                            \}. 
\end{align*}
%
%
Our algorithm computes, for each
subexpression $s$ of $s_0$,
$$\Theta(s)=\{ \gamma_{\Imc,\Dmc_s} \mid \Imc \text{ model of } \Dmc_s \text{
  and } \Omc \}.
$$
It then remains to check whether $\Theta(s_0)$ is non-empty.
The sets $\Theta(s)$ are computed as follows:
\begin{enumerate}[align=left]
\item[$\boldsymbol{s = \iota(\Dmc)\text{: }}$] Set $\Theta(s)=
  \{ \gamma_{\Imc,\Dmc} \mid \Imc \text{ model of } \Dmc \text{ and } \Omc
  \}$. 
\end{enumerate}


\begin{enumerate}[align=left]

\item[$\boldsymbol{s = s_1 \oplus s_2\text{: }}$] $\Theta(s)$ contains
  a model abstraction $\gamma=(T,F)$ for each pair of model abstractions
  $(\gamma_1, \gamma_2) \in \Theta(s_1) \times \Theta(s_2)$ with
  $\gamma_1^0 = \gamma_2^0$  defined by
  setting, for \mbox{$1 \leq i,j \leq k$},
  \begin{alignat*}{2}
    \gamma^T(i) &=\; &&\gamma^T_1(i) \cup \gamma^T_2(i) 
                      \\
    \gamma^E(i,j) &= &&\gamma^E_1(i,j) \cup \gamma^E_2(i,j) \, \cup\\
                 & &&\{ (t_1,\emptyset,t_2) \mid t_1 \in \gamma^T_{1}(i) 
                   \text{ and }
                    t_2 \in \gamma^T_{2}(j) \} \, \cup \\
                 & &&\{ (t_1,\emptyset,t_2) \mid t_1 \in \gamma^T_{2}(i) 
                   \text{ and }
                    t_2 \in \gamma^T_{1}(j) \}. 
  \end{alignat*}

\item[$\boldsymbol{s = \alpha_{i, j}^r(s')\text{: }}$] $\Theta(s)$
  contains a model abstraction $\gamma$ for every model
  abstraction $\widehat{\gamma} \in \Theta(s')$ such  that
  \begin{enumerate}


  \item[($*$)] for all $(t_1,M,t_2) \in {\widehat{\gamma}}^E(i,j)$, 
    there is a
    $t \in \mn{TP}_2$ such that $M \cup \{ r(x,y) \} \subseteq t$ and
    $t_1 \rightsquigarrow_t t_2$.

    
  \end{enumerate}
  The model abstraction $\gamma$ is defined like $\widehat{\gamma}$
  except that
  \begin{align*}
    &\text{ for all } (t_1, M, t_2) \in {\widehat{\gamma}}^E(i,j), \text{ add } r(x,y) \text{ to } M \text{, and }\\
    &\text{ for all } (t_1, M, t_2) \in {\widehat{\gamma}}^E(j,i), \text{ add } r(y,x) \text{ to } M\text{.}
  \end{align*}

\item[$\boldsymbol{s = \rho_{i \rightarrow j}(s')\text{: }}$] for each
  ${\widehat{\gamma}} \in \Theta(s')$, $\Theta(s)$ contains the model abstraction
  $\gamma$ defined like $\widehat{\gamma}$ except that,
  for \mbox{$1 \leq \ell \leq k$},
  \begin{alignat*}{2}
  \gamma^T(i) &= \emptyset  & \qquad  \gamma^E(i,\ell) &= \emptyset \\
  \gamma^T(j) &= {\widehat{\gamma}}^T(j) \cup {\widehat{\gamma}}^T(i) & \gamma^E(\ell,i) &= \emptyset\\
  \gamma^E(j,\ell) &= {\widehat{\gamma}}^E(j,\ell) \cup {\widehat{\gamma}}^E(i,\ell) \\
  \gamma^E(\ell,j) &= {\widehat{\gamma}}^E(\ell,j) \cup {\widehat{\gamma}}^E(\ell,i). 
  \end{alignat*}
\end{enumerate}
The algorithm achieves the intended goal and runs within the stated
time bounds.
\begin{restatable}{lemma}{gftwoaqalgocorrectnesslem}
  \label{lem:1exp_correctnes_gf2_aq}
  For all subexpressions $s$ of $s_0$,
  $\Theta(s)=\{ \gamma_{\Imc,\Dmc_s} \mid \Imc \text{ model of }
  \Dmc_s \text{ and } \Omc \}$.
\end{restatable}

\section{Lower Bounds}
\label{sect:lower}


We prove that for $(\text{GF}_2, \text{AQ})$ and $(\ALC,
\text{CQ})$, the double exponential running time in $|\Omc|$ cannot be
improved unless the exponential time hypothesis (ETH) is false. The
ETH states that if $\delta$ is the infimum of real numbers such that
3SAT can be solved in time $O(2^{\delta n})$, where $n$ is the number
of variables in the input formula, then $\delta > 0$.
The following are the two main results of this section.
\begin{restatable}{theorem}{twoexplowerthm}
\label{thm:2explower}
If the ETH is true, then there is no algorithm with running time
$2^{2^{o(|Q|)}} \cdot \mn{poly}(|\Dmc|)$ for OMQ
evaluation~in
\begin{enumerate}
\item $(\text{GF}_2, \text{AQ})$ on databases of cliquewidth 2;
\item $(\ALC, \text{CQ})$ on databases of cliquewidth 3.
\end{enumerate}
\end{restatable}
Note that OMQ evaluation in $(\text{GF}_2, \text{AQ})$ and 
$(\ALC, \text{CQ})$ are only \ExpTime-complete in combined
complexity on unrestricted databases. The former is folklore and the
latter is proved in \cite{DBLP:conf/dlog/Lutz08}.  Thus,
Theorem~\ref{thm:2explower} implies that, in these two OMQ languages,
one has to pay for the running time to be polynomial in $|\Dmc|$ on
databases of bounded cliquewidth by an exponential increase in the
running time in $|Q|$.

We now prove Theorem~\ref{thm:2explower}.
A major contribution of the original paper bringing forward the ETH is
the sparsification lemma, which implies the 
following useful consequence of the ETH \cite{ImpagliazzoPZ01}:
\begin{description}
\item[($*$)] If the ETH is true, then 3SAT cannot be solved in time $2^{o(m)}$, with $m$ the number of clauses of the input formula. 
\end{description}
Both results in Theorem~\ref{thm:2explower} are proved by
contradiction against ($*$), i.e.\ we show that an algorithm for these
OMQ evaluation problems with running time
$2^{2^{o(|Q|)}} \cdot \mn{poly}(|\Dmc|)$ can be used to design an algorithm
for 3SAT with running time $2^{o(m)}$, contradicting ($*$). To achieve
this, it suffices to give a polynomial time reduction from 3SAT to
OMQ evaluation that produces an ontology of size
$O(\log m)$ where $m$ is the number of clauses in the input formula. For both points of
Theorem~\ref{thm:2explower}, we first use a known reduction from 3SAT
to \textsc{List Coloring} (Theorem~\ref{lem:3sat-precol}) and then
reduce \textsc{List Coloring} to OMQ evaluation. In \textsc{List
  Coloring}, the input is a graph in which every
vertex is associated with a list of colors, and the question is
whether it is possible to assign to each vertex a color from its list
and obtain a valid coloring of the graph. Formally: Given an
undirected graph $G = (V, E)$ and a set $C_v \subseteq \mathbb{N}$ of
colors for each node $v\in V$, decide whether there is a map
$f : V \rightarrow \mathbb{N}$ such that $f(v) \in C_v$ for every
$v \in V$ and $f(u) \neq f(v)$ for every $\{u,v\} \in E$.

The following is proved in \cite{DBLP:tr/trier/MI96-41},
Theorem~11.  The graphs of cliquewidth~2 that are used in the proof
are (non-disjoint) unions of two cliques.
\begin{theorem}
\label{lem:3sat-precol}
There is a polynomial time reduction from \textsc{3SAT} to
\textsc{List Coloring} that turns a formula with $m$ clauses into an
instance of \textsc{List Coloring} with a connected graph of cliquewidth $2$ and $O(m)$ colors.
\end{theorem}
To prove Point~1 of Theorem~\ref{thm:2explower}, we reduce
\textsc{List Coloring} to OMQ evaluation in
$(\text{GF}_2, \text{AQ})$.  To highlight the idea of the reduction,
we first sketch the proof of a slightly weaker version of Point~1 of
Theorem~\ref{thm:2explower}, namely that there is no algorithm with
the stated running time on databases of cliquewidth~3 (instead of~2).
Moreover, we use queries of the form $\exists x \, B(x)$ rather than
AQs. As an intermediate, we consider the problem \textsc{Precoloring
  Extension}, where the input is a number of colors $k$ and a graph in
which some nodes are already colored, and the question is whether the
coloring can be extended to a $k$-coloring of the whole graph.
Consider the following standard reduction from \textsc{List Coloring}
to \textsc{Precoloring Extension} \cite{DBLP:journals/iandc/FellowsFLRSST11}.
For every vertex $v$ and every
color $c \notin C_v$, add a new vertex that is precolored with $c$ and
is only adjacent to $v$. This reduction increases the cliquewidth by
at most~$1$, so when starting from Theorem~\ref{lem:3sat-precol}, the
resulting graphs have cliquewidth at most 3. It remains to reduce
\textsc{Precoloring Extension} to OMQ evaluation.

Without the precoloring, we may view the input graph to
\textsc{Precoloring Extension} as a database~$\Dmc$ by replacing every
vertex by a constant and every edge $\{u,v\}$ by two facts $r(u,v)$
and $r(v,u)$.  Let us assume that the number of colors is a power
of~2, say $2^\ell$. Each of the $2^\ell$ colors is assigned a unique
bit string of length~$\ell$. Introduce concept names
$A_1^j, \ldots, A_\ell^j$ for $j \in \{0, 1\}$ and add the sentence
$\forall x ( A_i^0(x) \leftrightarrow \neg A_i^1(x))$ to the
ontology~\Omc, which says that every constant in $\Dmc$ should be
assigned a unique sequence of $\ell$ bits, representing a color. To
express that adjacent vertices should be colored differently, we just
have to say that their bit strings should differ in at least one
place:
$\forall x \forall y \, ( r(x,y) \rightarrow \bigvee_{i=1}^\ell
A_i^0(x) \leftrightarrow A_i^1(y))$. For every vertex $v$ that is
precolored by a color $c$, we add facts $A_i^j(v)$ that represent this
color. As the query, choose $\exists x \, B(x)$ with $B$ a fresh
concept name. This query is implied if and only if $\Dmc$ is not
satisfiable w.r.t.\ $\Omc$, which is the case if and only if we
started with a no-instance of \textsc{Precoloring Extension}.

The following lemma improves the reduction to use only databases of
cliquewidth~2 and AQs. 
\begin{restatable}{lemma}{listcolgftwolem}
\label{lem:listcol-gf2}
There is a polynomial time reduction from \textsc{List Coloring} on
connected graphs of cliquewidth~2 to OMQ evaluation in $(\text{GF}_2,
\text{AQ})$ on databases of cliquewidth~2, where the constructed
ontology is of size $O(\log (k))$ with $k$ the number of distinct colors of the input graph.
\end{restatable}

We are now ready to prove Point~1 of Theorem~\ref{thm:2explower}.

\smallskip

\noindent
\begin{proof}
  Assume for a contradiction that there is an algorithm with running
  time $2^{2^{o(|\Omc|)}} \cdot \mn{poly}(|\Dmc|)$ for OMQ evaluation in
  $(\text{GF}_2, \text{AQ})$ on databases of cliquewidth at most
  $2$. Then the following is an algorithm for 3SAT.
  %
  Given a formula in 3CNF with $m$ clauses, apply the reduction
  from Theorem~\ref{lem:3sat-precol} followed by the one from
  Lemma~\ref{lem:listcol-gf2}. This yields 
%
  an OMQ $Q \in (\text{GF}_2, \text{AQ})$ where the ontology
  \Omc is of size $O(\log(m))$, a database $\Dmc$ of cliquewidth $2$ and an answer candidate $a \in \mn{adom}(\Dmc)$.
  Use the algorithm for OMQ evaluation in $(\text{GF}_2, \text{AQ})$ whose existence we assumed on input $Q$, $\Dmc$ and $a$. 
The running time of the first step is polynomial. The running
time of the second step is $2^{2^{o(\log m)}} \cdot \mn{poly}(|\Dmc|)$, which
is~$2^{o(m)}$, contradicting the ETH.
\end{proof}
Interestingly, the reduction does not rely on the additional
expressive power of $\text{GF}_2$ compared to $\ALC$, but rather on
the ability of $\text{GF}_2$ to express certain statements more
succinctly. To demonstrate this, we introduce the extension $\ALC^=$ of
$\ALC$ with  the following constructor:

If $A_1,\ldots,A_n$ are concept names, then $\exists r.[A_1,\ldots,A_n]_{=}$ is a concept. The semantics is defined as follows:
\begin{align*}
\exists r.[A_1,\ldots,A_n]_{=}^\Imc  =  \{d \in \Delta^\Imc \mid \exists e \in \Delta^\Imc : (d,e) \in r^\Imc \text{ and }\\
d \in A_i^\Imc \Leftrightarrow e \in A_i^\Imc \text{ for all } i \in \{1, \ldots, n\}\}.
\end{align*}
It is easy to see that the expressive power of \ALC and of $\ALC^=$
are identical, as the additional constructor
$\exists r.[A_1,\ldots,A_n]_{=}$ can be expressed as
$$
   \bigsqcup_{C_1 \in \{A_1, \neg A_1\}} \!\!\!\!
   \cdots \!\!\!\!
   \bigsqcup_{C_n \in \{A_n, \neg A_n\}} \!
   \big (
   C_1 \sqcap \cdots \sqcap C_n \sqcap \exists r.
   (C_1 \sqcap \cdots \sqcap C_n)
   \big ).
   $$
   Note that the above concept is exponentially larger than the original
   one. In GF$_2$, in contrast, the new constructor can easily be
   expressed using only a linear size formula.

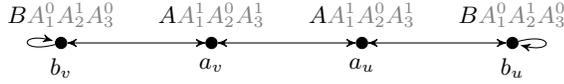
\begin{figure}[t]
\begin{center}
\begin{tikzpicture}[->,>=stealth',font=\sffamily\small]
\tikzstyle{node}=[shape=circle, draw,inner sep=1.5pt, fill=black]
\node[node, label=below:$b_v$, label=above:$B${\color{gray} $A_1^0A_2^1A_3^0$}] (BV) at (0,0) {}; 
\node[node, label=below:$a_v$, label=above:$A${\color{gray} $A_1^1A_2^0A_3^1$}] (AV) at (2,0) {}; 
\node[node, label=below:$a_u$, label=above:$A${\color{gray} $A_1^1A_2^0A_3^1$}] (AU) at (4,0) {}; 
\node[node, label=below:$b_u$, label=above:$B${\color{gray} $A_1^0A_2^1A_3^0$}] (BU) at (6,0) {}; 
\draw [<->] (BV) -- (AV) node[midway,above left] {}; 
\draw [<->] (AV) -- (AU) node[midway,above] {}; 
\draw [<->] (AU) -- (BU) node[midway,above] {}; 
\draw [->] (BV) edge[loop left] (BV); 
\draw [->] (BU) edge[loop right] (BU); 
\end{tikzpicture}
\caption{Black: Fragment of \Dmc that represents edge $\{u,v\} \in E$. 
Gray: Concept names that could appear in an interpretation where the coloring has a defect, 
for $\ell = 3$.}
\label{fig:cq-lower-db}
\end{center}
\vspace*{-6mm}
\end{figure}
In the reduction used to prove Lemma~\ref{lem:listcol-gf2}, we may
replace the $\text{GF}_2$ ontology by the following $\ALC^=\!$-ontology:
\begin{align*}
\Omc = \, &\{ \exists s.[A_1,\ldots,A_n]_{=} \sqsubseteq D \} \,\cup\\[0.5mm]
&\{ B \sqcap \exists r.[A_1,\ldots,A_n, B]_{=} \sqsubseteq D \} \,\cup
  \\[1mm]
&\{ \exists r.D \sqsubseteq D, D \sqsubseteq \forall r.D, \exists s.D \sqsubseteq D, D \sqsubseteq \forall s.D  \} \,\cup  
\end{align*}
\begin{align*}
&\{\top \sqsubseteq A_i^0 \sqcup A_i^{1}, A_i^0 \sqcap A_i^{1}
\sqsubseteq \bot \mid i \in \{1, \ldots, \ell\}\}.
\end{align*} 
Thus, Point~1 of
Theorem~\ref{thm:2explower} still holds if $(\text{GF}_2,\text{AQ})$
is replaced with $(\ALC^=, \text{AQ})$.

To prove Point~2 of Theorem~\ref{thm:2explower}, we use a reduction
that is similar to the one presented before
Lemma~\ref{lem:listcol-gf2}. 
This time we are only allowed to use an $\ALC$ ontology, but also
CQs in place of AQs, which again allows us to express succinctly
that adjacent nodes have to be colored differently.
\begin{restatable}{lemma}{listcolalccqlem}
\label{lem:listcol-alc-cq}
There is a polynomial time reduction from \textsc{List Coloring} on
connected graphs of cliquewidth~2 to OMQ evaluation in
$(\ALC, \text{CQ})$ on databases of cliquewidth~3, where the
constructed OMQ is of size $O(\log k)$, $k$ the number of
distinct colors of the input graph.
\end{restatable}

\noindent
\begin{proof} 
Let $G=(V, E)$ be an undirected graph of cliquewidth~2
and $C_v \subseteq \mathbb{N}$ a list of possible colors for every $v \in V$.
We construct a database $\Dmc$ of cliquewidth~3 and a Boolean
OMQ $Q=(\Omc,q)$ from $(\ALC, \text{CQ})$ 
such that $G$ and $C_v$ are a yes-instance for \textsc{List Coloring}
if and only if $\Dmc \not \models Q$.
In contrast to what was described above, we do not carry out a
reduction
from {\sc List Coloring} to {\sc Precoloring Extension} as a separate
first step, but build it directly into the main reduction.
Let $k$ be the number of distinct colors that occur, w.l.o.g.\ let $\bigcup_{v \in V} C_v = \{1, \ldots, k\}$.
Define $\ell = \lceil \log k \rceil$ and $L = \{1, \ldots, 2^\ell\}$.

We use concept names $A_i^j$,
$i \in \{1, \ldots, \ell\}$ and $j \in \{0, 1\}$, another concept name
$B$, one role name~$r$, and constants
$a_v, b_v, a_v^c$, and $b_v^c$ for every $v \in V$ and
$c \in L \setminus C_v$.  The database is defined as 
\begin{align*}
\Dmc = & \{r(a_v, a_u), r(a_u, a_v) \mid \{u, v\} \in E\}\, \cup\\ 
& \{r(a_v, a_v^c), r(a_v^c, a_v) \mid v \in V, c \in L \setminus C_v\}\, \cup\\
& \{A_i^j(a_v^c) \mid v \in V, c \in L \setminus C_v \text{ and the } i \text{-th bit of } c \text{ is } j\}\,\cup\\ 
& \{r(a_v, b_v), r(b_v, a_v), r(b_v, b_v) \mid v \in V\}\,\cup\\ 
& \{r(a_v^c, b_v^c), r(b_v^c, a_v^c), r(b_v^c, b_v^c) \mid v \in V, c \in L \setminus C_v\} \, \cup\\ 
& \{A(a_v), A(a_v^c) \mid v \in V, c \in L \setminus C_v\} \, \cup\\ 
& \{B(b_v), B(b_v^c) \mid v \in V, c \in L \setminus C_v\}.
\end{align*}
Here, every constant $a_v$ represents a vertex of $G$, and every constant $a_v^c$ is a
constant adjacent to $a_v$ that is precolored with the color $c$. Furthermore, for every constant
$a_v$ (resp.\ $a_v^c$) there is a constant $b_v$ (resp.\ $b_v^c$) such that $b_v$ (resp. $b_v^c$) is
adjacent only to $a_v$ (resp.\ $a_v^c$).
We refer to the $b_v$ and $b_v^c$ as the \emph{copies} of $a_v$ and $a_v^c$.
Every copy has an $r$-selfloop and a unary marker $B$,
and every non-copy has a unary marker~$A$.
Figure~\ref{fig:cq-lower-db} depicts a part of the database.
We argue in the appendix that $\Dmc$ has cliquewidth~3.
\begin{figure}[t]
\begin{center}
\begin{tikzpicture}[->,>=stealth',font=\sffamily\small]
\tikzstyle{node}=[shape=circle, draw,inner sep=1.5pt, fill=black]
\node[node, label=below:$x$, label=above:$B$] (X) at (0,0) {}; 
\node[node, label=below:$z_1^i$, label=above:$A_i^0$] (Z1) at (1,1) {}; 
\node[node, label=below:$z_2^i$, label=above:$A$] (Z2) at (2,1) {}; 
\node[node, label=below:$z_3^i$, label=above:$A_i^1$] (Z3) at (3,1) {}; 
\node[node, label=below:$y$, label=above:$B$] (Y) at (4,0) {}; 
\node[node, label=below:$z_4^i$, label=above:$A_i^1$] (Z4) at (1,-1) {}; 
\node[node, label=below:$z_5^i$, label=above:$A$] (Z5) at (2,-1) {}; 
\node[node, label=below:$z_6^i$, label=above:$A_i^0$] (Z6) at (3,-1) {}; 
\draw [->] (X) -- (Z1) node[midway,above left] {}; 
\draw [->] (Z1) -- (Z2) node[midway,above] {}; 
\draw [->] (Z2) -- (Z3) node[midway,above] {}; 
\draw [->] (Z3) -- (Y) node[midway,above right] {}; 
\draw [->] (X) -- (Z4) node[midway,above right] {}; 
\draw [->] (Z4) -- (Z5) node[midway,above] {}; 
\draw [->] (Z5) -- (Z6) node[midway,above] {}; 
\draw [->] (Z6) -- (Y) node[midway,above left] {}; 
\draw [->] (X) edge[loop left] (X); 
\draw [->] (Y) edge[loop right] (Y); 
\end{tikzpicture}
\caption{Gadget $i$ of CQ $q$. Every edge represents an $r$-fact. The variables $x$ and $y$
are shared among gadgets.}
\label{fig:cq-lower}
\end{center}
\vspace*{-6mm}
\end{figure}
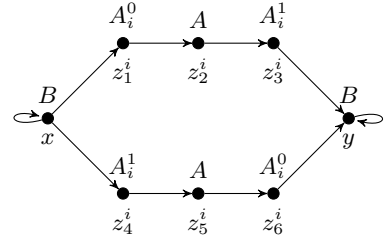

Next, we construct the ontology, which has two purposes. First, we assign a unique color to every
constant of the database. For every $i \in \{1, \ldots, \ell\}$, we
include the CIs
$\top \sqsubseteq A_i^0 \sqcup A_i^1$ and $A_i^0 \sqcap A_i^1 \sqsubseteq \bot$.
The second purpose is to assure that every copy $b_v$ (resp.\ $b_v^c$) is
colored using the anti-color of the color of $a_v$ (resp.\ $a_v^c$).
Here, the \emph{anti-color} of a color $c$
is the color obtained from $c$ by flipping every bit, e.g.\ the anti-color of $10010$ is $01101$.
Since the copies $b_v$ and $b_v^c$ carry the concept name~$B$, this can be achieved using
the following CIs for every $i \in \{1, \ldots, \ell\}$ and $j \in \{0, 1\}$:
$$
A \sqcap B \sqsubseteq \bot \qquad
B \sqcap A_i^j  \sqsubseteq \forall r. (B \sqcup A_i^{1-j}).
$$
Now we construct the (Boolean) CQ $q$, whose purpose it is to
detect a defect in the coloring. It is a union of $\ell$ gadgets,
with gadget number $i$ checking that
the colors of two adjacent vertices agree on bit number $i$.  We use
two variables $x$ and $y$, as well as variables
$z_1^i, z_2^i, \ldots, z_6^i$ for every $i \in \{1, \ldots, \ell\}$.
The variables $x$ and $y$ are shared between all the gadgets, whereas
the variables $z_1^i, z_2^i, \ldots, z_6^i$ belong to gadget number~$i$.  The $i$-th gadget is displayed in Figure~\ref{fig:cq-lower}.
It is clear that the OMQ $(\Omc,q)$ is of size $O(\log k)$. 
%
%
\\[2mm]
\emph{Claim.} $G$ and $C_v$ are a yes-instance for \textsc{List Coloring}
if and only if $\Dmc \not\models Q$.
\\[2mm]
We only show the `if' direction here as it demonstrates how the
constructed CQ $q$ can identify defects in the coloring. 
Let $\Imc$ be a model of $\Omc$
and $\Dmc$ such that $\Imc \not \models q$.  We show that $G$ and the
$C_v$ are a yes-instance of \textsc{List Coloring}. Since $\Imc$ is a
model of $\Omc$, every $a_v$ is assigned a unique color $c_v$, encoded
using the $A_i^j$. We define $f(v) = c_v$ for every $v \in V$. If any
two adjacent constants $a_u$ and $a_v$ were colored with the same
color $c = j_1 j_2 \ldots j_\ell \in \{0, 1\}^\ell$, then we can
construct a homomorphism $h$ from $q$ to $\Imc$ by
$h(x) = b_v, h(y) = b_u$ and, if $j_i = 0$, then
\begin{align*}
&h(z_1^i) = a_v, h(z_2^i) = a_u, h(z_3^i) = b_u\\
&h(z_4^i) = b_v, h(z_5^i) = a_v, h(z_6^i) = a_u
\end{align*}
and if $j_i = 1$, then
\begin{align*}
&h(z_1^i) = b_v, h(z_2^i) = a_v, h(z_3^i) = a_u\\
&h(z_4^i) = a_v, h(z_5^i) = a_u, h(z_6^i) = b_u\,.
\end{align*}
Figure~\ref{fig:cq-lower-db} depicts the case where $c=101$.
But we assumed that $\Imc \not \models q$, so every two adjacent $a_v$ and $a_u$ have
different colors. Thus, $f$ witnesses that $G$ and the $C_v$ are a yes-instance for
\textsc{List Coloring}. 
\end{proof}
%
%
With Lemma~\ref{lem:listcol-alc-cq} at hand, the proof for Point~2 of Theorem~\ref{thm:2explower}
continues exactly as the proof of Point~1.


\section{Bounded Treewidth}
\label{sect:btw}

We prove two results that concern OMQ evaluation on databases of
bounded treewidth, mainly to contrast with the results obtained for
databases of bounded cliquewidth.
\begin{restatable}{theorem}{thmtwone}
  \label{thmtwone}
  In $(\text{GF}_2,\text{AQ})$, an OMQ $Q=(\Omc,q)$ can be evaluated on
  a database \Dmc of treewidth $k$ in time
  \mbox{$2^{O(|\Omc| \cdot k^2)} \cdot |\Dmc|$}.
\end{restatable}
Note that we obtain single exponential running time while
Theorem~\ref{thm:2explower} states that when treewidth is
replaced by cliquewidth, we cannot be better than double exponential.
\begin{restatable}{theorem}{thmtwtwo}
  \label{thmtwtwo}
    In $(\text{GF}_2,\text{UCQ})$, evaluating an OMQ $Q=(\Omc,q)$ on
  a database \Dmc of treewidth $k$ is possible in time
  $2^{|\Omc| \cdot k^{O(|q| \log(|q|))}} \cdot |\Dmc|$.
\end{restatable}
Theorem~\ref{thmtwtwo} is relevant because we leave open the
corresponding case for cliquewidth, that is, we do not even know
whether OMQ evaluation in $(\text{GF}_2,\text{UCQ})$ on databases of
bounded cliquewidth is in \PTime in data complexity. 
We remark that
the overall running time in Theorem~\ref{thmtwtwo} cannot be improved
to single exponential because OMQ evaluation in $(\ALCI,\text{CQ})$
is \TwoExpTime-complete in combined complexity on databases
of the form $\{ A(c) \}$ \cite{DBLP:conf/dlog/Lutz07}.

We give a single algorithm that yields both Theorem~\ref{thmtwone}
and~\ref{thmtwtwo}.  Assume that we are given as input an OMQ
$Q=(\Omc,q) \in (\text{GF}_2, \text{UCQ})$, and a database $\Dmc_0$ of
treewidth~$k$.  As in the proof of Theorem~\ref{thm:alciucq2exp}, we
may assume w.l.o.g.\ that $Q$ is Boolean.
In contrast to the case of cliquewidth, we do not need to assume that
a tree decomposition is given as part of the input. In fact, a tree
decomposition $T=(V,E, (B_v)_{v \in V})$ of $\Dmc_0$ of width at most
$2k+1$ can be computed in time $2^{O(k)} \cdot
|\mn{adom}(\Dmc_0)|$ 
\cite{korhonen2021single}. 

Our algorithm traverses the tree
decomposition $T$ bottom-up. For every node $v$, we compute a representation of
the models of $\Dmc_v$ and \Omc, where $\Dmc_v$ is the restriction of
$\Dmc_0$ to the constants that occur in the bags of the subtree of $T$
rooted at $v$. 
Along with the
representations of models, we keep track of partial homomorphisms from
CQs in $q$ to these models (we actually rewrite the query beforehand,
in the style of Section~\ref{sect:alciucq}). 
Details are in the appendix.

\section{Conclusion}
\label{sect:concl}

  We leave open the interesting question whether OMQ evaluation in
  $(\text{GF}_2,\text{(U)CQ})$ on databases of bounded cliquewidth is
  in \PTime in data complexity.  While all
  other OMQ languages studied in this paper are easily translated into
  MSO$_1$, this is not the case for
  $(\text{GF}_2,\text{(U)CQ})$. In fact, it is straightforward to
  express the problem \textsc{Monochromatic Triangle}\footnote{Can
    the edges of a given graph be colored with two colors without generating a
    monochromatic triangle?}
  in $(\text{GF}_2,\text{CQ})$ and to our knowledge it is open whether
  this problem can be expressed in MSO$_1$. 

  It would also be interesting to consider natural
  extensions of \ALCI. We believe that adding role inclusions has no
  impact on the obtained results. It is less clear what
  happens for transitive roles, number restrictions, and nominals. In
  the appendix, we make the following observation for the extension
  $\mathcal{ALCF}$ of \ALC with (globally) functional roles.
  \begin{restatable}{theorem}{thmalcf}
    \label{thm:alcf}
    OMQ evaluation in $(\mathcal{ALCF},\text{CQ})$ on databases of
    treewidth~2 is \coNP-hard in data complexity if the unique
    name assumption is not made. 
  \end{restatable}
  The proof relies on an observation from
  \cite{figueira2016semantically}. We conjecture that when the unique
  name assumption is made, fixed-parameter linearity is regained and
  can be proved using the methods in this paper.

  It would also be very interesting to pursue the idea of
  decomposing databases into a (hopefully small) part of high
  cliquewidth and parts of low cliquewidth, with the aim of achieving
  efficiency in practical applications.

\section*{Acknowledgements}
  Supported by the DFG CRC 1320 EASE - Everyday Activity Science and Engineering.
  
\bibliographystyle{kr}
\bibliography{literature}

\cleardoublepage 

\appendix

\section{Additional Preliminaries}

{\bf $O$-Notation.} The use of the $O$-notation is somewhat subtle
when multiple inputs/variables and exponential functions are involved.
Regarding the multiple inputs, we follow the standard convention that
is mentioned in \cite{cormen2022introduction} and actually rather well explained on the
Wikipedia page on ``Big $O$ Notation''. Assume that we have three
variables $n,m,k$. When writing something like
$f(n,m,k) \in n \cdot 2^{m \cdot 2^{O(k)}}$ what we mean is
$$\log \log f(n,m,k) \in O(\log \log n \cdot \log m \cdot k).$$ As a
convention, we only use the $O$-notation for the highest exponent,
corresponding to the intuition that multiplicative factors are more
`effective' there (given that none of our input sizes/variables can be
zero).

\smallskip
{\bf Homomorphisms and the semantics of CQs.} 
A \emph{homomorphism} from interpretation $\Imc_1$ to interpretation 
$\Imc_2$ is a function $h :\Delta^{\Imc_1} \to \Delta^{\Imc_2}$ such 
that $d \in A^{\Imc_1}$ implies $h(d) \in A^{\Imc_2}$ and 
$(d,e) \in r^{\Imc_1}$ implies $(h(d),h(e)) \in r^{\Imc_2}$ for all 
$d,e \in \Delta^{\Imc_1}$, $A \in \NC$, and $r \in \NR$. 

A CQ $q$ gives rise to a database $\Dmc_q$, often called the
\emph{canonical database} for $q$, obtained by viewing the variables
in $q$ as constants and the atoms as facts.  A homomorphism from CQ
$q$ to interpretation \Imc is a homomorphism from $\Dmc_q$ to \Imc. As
state in the main body of the paper, a
tuple $\bar d \in (\Delta^\Imc)^{|\bar x|}$
is then an \emph{answer} to $q$ on
\Imc 
if there is a homomorphism $h$
from~$q$ to \Imc with $h(\bar x)=\bar d$. 

\smallskip {\bf Tree Interpretation.} 
An interpretation \Imc is a \emph{tree}
if the undirected graph
$G_\Imc = (V,E)$ with $V=\Delta^\Imc$ and
$E=\{ \{d,e\} \mid (d,e) \in r^\Imc \text{ for some } r \in \NR \}$ is
a tree and there are no self loops and multi-edges.

\smallskip {\bf Treewidth.}  Treewidth is a widely used notion that
measures the degree of tree-like\-ness of a graph.
A \emph{tree decomposition} of a database $\Dmc$ is a triple
$(V,E, (B_v)_{v \in V})$ where $(V, E)$ is an undirected tree and
$(B_v)_{v \in V} $ is a family of subsets of $\mn{adom}(\Dmc)$, often
referred to as \emph{bags}, such that:
\begin{enumerate} 
\item for all $c \in \mn{adom}^\Dmc$, $\{v \in V\mid c \in B_v\}$ is
  nonempty and connected in~$(V,E)$;
\item for every $r(c_1,c_2)\in \Dmc$, there is a $v \in V$ with $c_1,c_2 \in B_v$.
\end{enumerate}
The \emph{width} of $(V,E, (B_v)_{v \in V})$ is $\min\{ |B_v| \mid v
 \in V\}-1$.
%
$|B_v \cap B_{v'}| \leq \ell$ and 
$|B_v| \leq k$. 
A database \Dmc \emph{has treewidth} $k$ if it has tree decomposition
of width $k$, but not of width~$k-1$.

We also speak about the treewidth of CQs and of interpretations.  Each
CQ $q$ gives rise to a database $\Dmc_q$, often called the canonical
database for $q$, obtained by viewing the variables in $q$ as
constants and the atoms as facts. The treewidth of $q$ is that of
$\Dmc_q$. Likewise, an interpretation \Imc gives rise to the
(potentially infinite) database $\Dmc_\Imc = \{ A(d) \mid d \in A^\Imc
\}
\cup \{ r(d,e) \mid (d,e) \in r^\Imc\}$ and the treewidth of \Imc is
that of $\Dmc_\Imc$.

\section{More on Example~\ref{ex:teacher_pupils_schools}}

We first make precise the class of databases that we have described
informally in Example~\ref{ex:teacher_pupils_schools}. To obtain a
database from the class, we choose a number
of schools~$n_s$, a number of teachers $n_t$, numbers of pupils
$n^{(1)}_p,\dots,n^{(n_t)}_p$, and a surjective assignment of teachers
to schools $w:\{1,\dots,n_t\} \rightarrow \{1,\dots,n_s\}$.
The associated database then consists of the following facts:
\begin{itemize}

\item $\mn{School}(s_i)$ for $1 \leq i \leq n_s$;

\item $\mn{Teacher}(t_i)$ for $1 \leq i \leq n_t$;

\item $\mn{Pupil}(p_{i,j})$ for $1 \leq i \leq n_t$ and $1 \leq j \leq
  n^{(i)}_p$;

\item $\mn{teaches}(t_i,p_{i,j})$ for $1 \leq i \leq n_t$ and
  $1 \leq j \leq n^{(i)}_p$;

\item $\mn{worksAt}(t_i,s_{w(i)})$ for $1 \leq i \leq n_t$.
  
\end{itemize}
This is illustrated in Figure~\ref{fig:example_class_of_databases_structure}.
\newcommand\single[2]{ 
  \foreach \x in {1,...,#2}{
  \pgfmathsetmacro{\ang}{360/#2}
      \pgfmathparse{(\x-1)*\ang}
      \node[draw,fill=black,circle,inner sep=0.7pt] (#1-\x) at (\pgfmathresult:0.25cm) {};
    }
    \foreach \x [count=\xi from 1] in {1,...,#2}{
      \foreach \y in {\x,...,#2}{
      \path (#1-\xi) edge[-] (#1-\y);
    }
  }
}
\begin{figure}
  \centering
  \begin{tikzpicture}
    \begin{scope}[shift={(-3,0)}]    
      \node [draw, circle, inner sep=8pt, label={[label distance=0cm, scale=0.8]180:$P_1$}] (c1) at (0,0) {};
      \single{a}{3}
      \node[draw,circle,inner sep=0.7pt] (t1) at (0, 1.4) {$t_1$};
    \end{scope}
    \node [draw, circle, inner sep=8pt, label={[label distance=0cm, scale=0.8]180:$P_2$}] (c2) at (0,0) {};
    \node[draw,circle,inner sep=0.7pt] (t2) at (0, 1.4) {$t_2$};
    \node[draw,circle,inner sep=0.7pt] (w1) at (-1.5, 2.7) {$s_1$};
    \single{b}{4}
    \begin{scope}[shift={(2,0)}]   
      \node [draw, circle, inner sep=8pt, label={[label distance=0cm, scale=0.8]180:$P_3$}] (c3) at (0,0) {};
      \single{a}{4}
      \node[draw,circle,inner sep=0.7pt] (t3) at (0, 1.4) {$t_3$};
      \node[draw,circle,inner sep=0.7pt] (w2) at (0, 2.7) {$s_2$};
    \end{scope}
    \draw[->] (t1) -- node[auto, sloped, scale=0.5] {teaches} (c1);
    \draw[->] (t2) -- node[auto, sloped,scale=0.5] {teaches} (c2);
    \draw[->] (t3) -- node[auto, sloped,scale=0.5] {teaches} (c3);
    \draw[->] (t1) -- node[auto, sloped,scale=0.5] {worksAt} (w1) ;
    \draw[->] (t2) -- node[auto, sloped,scale=0.5] {worksAt}  (w1);
    \draw[->] (t3) -- node[auto, sloped,scale=0.5] {worksAt}  (w2);
  \end{tikzpicture}
  \caption{Structure of the considered class of databases.}
  \label{fig:example_class_of_databases_structure}
\end{figure}
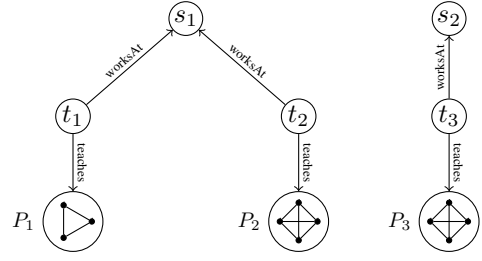
All databases from this class have cliquewidth~3.
Figure~\ref{fig:solution_for_ex_1} presents a 3-expression that
generates the concrete example database given in the main body of the
paper, where 
    \begin{align*}
      \Dmc_1 &= \{\mn{Pupil}(a_1), L_1(a_1)\} 
      & \Dmc_2 &= \{\mn{Pupil}(a_2), L_2(a_2)\} \\
      \Dmc_3 &= \{\mn{Teacher}(b), L_2(b)\} 
      & \Dmc_4 &= \{\mn{School}(c), L_3(c)\} 
    \end{align*}
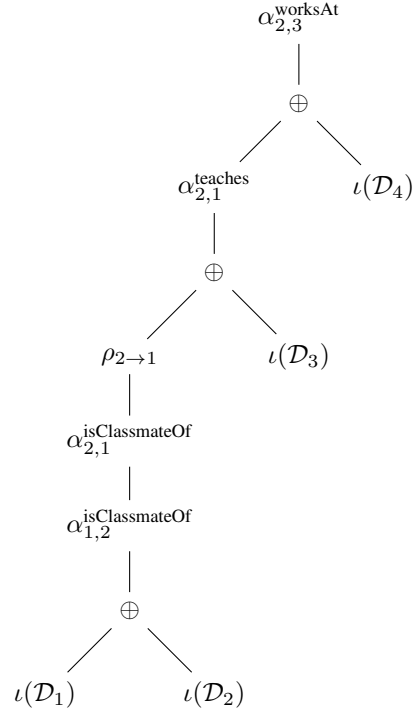
\begin{figure}[t]
  \centering
    \begin{tikzpicture}[align=center, scale=0.75]
    \node {$\alpha_{2, 3}^\text{worksAt}$}
    child {node{$\oplus$}
        child {node {$\alpha_{2, 1}^{\text{teaches}}$}
          child {node{$\oplus$}
            child {node{$\rho_{2 \rightarrow 1}$}
              child {node{$\alpha_{2, 1}^{\text{isClassmateOf}}$}
                child {node{$\alpha_{1, 2}^{\text{isClassmateOf}}$}
                  child {node{$\oplus$}
                    child {node{$\iota(\Dmc_1)$}}
                    child[missing] {}
                    child {node{$\iota(\Dmc_2)$}}
                  }
                }
              }
            }
            child[missing] {}
            child {node{$\iota(\Dmc_3)$}}
          }
        }
      child[missing] {}
      child {node {$\iota(\Dmc_4)$}}
    };
  \end{tikzpicture}
  \caption{$3$-expression generating the database in Example
    \ref{ex:teacher_pupils_schools}.}
    \label{fig:solution_for_ex_1}
\end{figure}  
In general, one starts with creating the first group of pupils along
with their teacher using a 2-expression, then the second group along
with their teacher, for all teachers that work at the same school. One
then takes a disjoint union, introduces the school (with a third
label) that all generated teachers work at, and puts the required
\mn{worksAt} edges. After that, one proceeds in the same way with the
second school, and so forth, at the end taking a disjoint union.

\section{Proofs for Section~\ref{sect:upper}}
\subsection{\ALCI with AQs}

\alciaqublem*
\noindent
\begin{proof}
%
  We prove by induction on the structure of $s$ that the equality in Lemma \ref{lem:1exp-correctness} holds.
  For the base case $s = \iota(\Dmc)$ it holds per definition.
  There are three cases in the induction step that we further split into the ``$\subseteq$'' and ``$\supseteq$'' direction.

  $s = s_1 \oplus s_2$, ``$\subseteq$'': Let $\gamma \in \Theta(s)$ be
  an IOA. We have to identify a model \Imc of $\Dmc_s$ and \Omc such
  that $\gamma=\gamma_{\Imc,\Dmc_s}$.  Since $\gamma \in \Theta(s)$,
  there are $\gamma_1 \in \Theta(s_1)$ and $\gamma_2 \in \Theta(s_2)$
  that result in $\gamma$ being included in $\Theta(s)$.  By the induction
  hypothesis, we find models $\Imc_i$ of $\Dmc_{s_i}$ and \Omc
  such that
  $\gamma_i = \gamma_{\Imc_i,\Dmc_{s_i}}$ for $i \in \{1,2\}$.  We can
  assume the domains of $\Imc_1$ and $\Imc_2$ to be disjoint because
  the active domains of $\Dmc_{s_1}$ and $\Dmc_{s_2}$ are disjoint.
  Choose as \Imc the union of $\Imc_1$ and $\Imc_2$, which  is a model of
  $\Dmc_s$ and \Omc. Moreover, $\mn{tp}_\Imc(c)= \mn{tp}_{\Imc_i}(c)$
  for all $c \in \mn{adom}(\Dmc_{s_i})$, $i \in \{1,2\}$. Thus, the
  definition of the IOA $\gamma_{\Imc,\Dmc_s}$ and the construction of
  $\gamma$ from $\gamma_1=\gamma_{\Imc_i,\Dmc_{s_1}}$ and
  $\gamma_2=\gamma_{\Imc_i,\Dmc_{s_2}}$ imply that
  $\gamma_{\Imc,\Dmc_s}=\gamma$, as required.

  $s = s_1 \oplus s_2$, ``$\supseteq$'': Let $\Imc$ be a model of $\Dmc_s$
  and $\Omc$. Then $\Imc$ is also a model of $\Dmc_{s_1}$ and
  $\Dmc_{s_2}$. Thus by the induction hypothesis,
  $\gamma_{\Imc,\Dmc_{s_i}} \in \Theta(s_i)$ for $i \in \{1,2\}$.
  What remains to be shown is that this results in $\gamma_{\Imc, \Dmc_{s}}$
  to be included in $\Theta(s)$. We may argue similarly to the
  previous case.

  $s = \alpha_{i, j}^r(s')$, ``$\subseteq$'':
  Let $\gamma \in \Theta(s)$ be an IOA.
  We have to identify a model \Imc of $\Dmc_s$ and \Omc such that $\gamma=\gamma_{\Imc,\Dmc_s}$.
  The IOA $\gamma$ is also in $\Theta(s')$ because our algorithm only removes IOAs in the construction of $\Theta(s)$.
  By the induction hypothesis, we find a model $\Imc'$ of $\Dmc_{s'}$ and $\Omc$ such that $\gamma_{\Imc', \Dmc_{s'}} = \gamma$.
  Let $\Imc = \Imc' \cup \{r(a,b) \mid L_i(a) \in \Dmc_{s'} \text{ and } L_j(b) \in \Dmc_{s'}\}$.
  Then $\Imc$ is a model of $\Dmc_s$.
  Moreover, $\mn{tp}_\Imc(c)= \mn{tp}_{\Imc'}(c)$ for all $c \in \mn{adom}(\Dmc_{s})$ because $\gamma$ satisfies the two conditions in the $\Theta(s)$ construction.
  In fact this can be shown by a straightforward induction on the structure of concepts $C \in \mn{cl}(\Omc)$: 
  $C \in \mn{tp}_\Imc(c)$ iff. $C \in \mn{tp}_{\Imc'}(c)$ for all $c \in \mn{adom}(\Dmc_{s})$.
  The most interesting case in this induction is where $\neg \exists r.D \in \mn{tp}_{\Imc'}(c)$.
  For this consider that we also have the NNF of all concepts in our types, so $\forall r.\bar D \in \mn{tp}_{\Imc'}(c)$, 
  with $\bar D$ being $\neg D$ in NNF. So now $\bar D \in \gamma^{\mn{out}}$ per induction hypothesis and thus rule $a)$ prevents $D$
  from being satisfied by any constant labeled $j$.
  Thus $\Imc$ is also a model of $\Omc$ and $\gamma_{\Imc,\Dmc_s} = \gamma_{\Imc', \Dmc_{s'}} = \gamma$ as required.

  $s = \alpha_{i, j}^r(s')$, ``$\supseteq$'':
  Let $\Imc$ be a model of $\Dmc_s$ and $\Omc$.
  Then $\Imc$ is also a model of $\Dmc_{s'}$.
  By the induction hypothesis, $\gamma_{\Imc, \Dmc_{s'}} \in \Theta(s')$.
  Thus $\gamma_{\Imc, \Dmc_{s'}} = \gamma_{\Imc, \Dmc_{s}}$ and what remains to be shown is that $\gamma_{\Imc, \Dmc_{s'}}$ does not get removed in the construction of $\Theta(s)$.
  This can be shown by a straightforward induction on the structure of concepts $C \in \mn{cl}(\Omc)$: if $C$ satisfies the ``if'' part of one condition, in the construction of $\gamma_{\Imc, \Dmc_{s}} \in \Theta(s)$ from $\gamma_{\Imc, \Dmc_{s'}} \in \Theta(s')$, then the ``then'' part is also satisfied. 

  $s = \rho_{i \rightarrow j}(s')$, ``$\subseteq$'':
  Let $\gamma \in \Theta(s)$ be an IOA.
  We have to identify a model \Imc of $\Dmc_s$ and \Omc such that $\gamma=\gamma_{\Imc,\Dmc_s}$.
  Since $\gamma \in \Theta(s)$ there is a $\widehat{\gamma} \in \Theta(s')$ that results in $\gamma$ being included in $\Theta(s)$.
  By the induction hypothesis, we find a model $\Imc'$ of $\Dmc_{s'}$ and $\Omc$ such that $\gamma_{\Imc', \Dmc_{s'}} = \widehat{\gamma}$.
  Let $\Imc$ be just like $\Imc'$ except $L_i^\Imc = \emptyset$ and $L_j^\Imc = L_i^{\Imc'} \cup L_j^{\Imc'}$.
  Thus, $\Imc$ is a model of $\Dmc_s$ and $\Omc$ and the definition of the IOA $\gamma_{\Imc,\Dmc_s}$ and the construction of $\gamma$ from $\widehat{\gamma}$ imply that $\gamma_{\Imc,\Dmc_s} = \gamma$, as required.

  $s = \rho_{i \rightarrow j}(s')$, ``$\supseteq$'':
  Let $\Imc$ be a model of $\Dmc_s$ and $\Omc$.
  Moreover let $\Imc'$ be just like $\Imc$ except $L_i^{\Imc'} = \{a \mid L_i(a) \in \Dmc_{s'}\}$ and $L_j^{\Imc'} = L_j^{\Imc} \setminus L_i^{\Imc'}$.
  Then $\Imc'$ is a model of $\Dmc_{s'}$ and $\Omc$.
  Thus by the induction hypothesis, $\gamma_{\Imc', \Dmc_{s'}} \in \Theta(s')$.
  What remains to be shown is that this results in $\gamma_{\Imc, \Dmc_{s}}$ to be included in $\Theta(s)$. 
  We can argue similarly to the previous case.
\end{proof}
\alciaqubthm*
\noindent
\begin{proof}
  Correctness of the algorithm follows immediately from Lemma
  \ref{lem:1exp-correctness}.  Next we analyze the running time.

  For every subexpression $s$ of $s_0$, we construct a set of IOAs
  $\Theta(s)$. Every IOA is of size $O(|\Omc| \cdot k)$ and there are
  at most $2^{O(|\Omc| \cdot k)}$ IOAs.  In the construction of
  $\Theta(s_0)$, we do a bottom-up walk over $s$, visiting every
  subexpression $s_0$ only once. By assumption, the number of such
  subexpressions is $O(|\Dmc|)$. The computation of each $\Theta(s)$
  is clearly possible in time $2^{O(|\Omc| \cdot k)}$.  Let us only
  mention that in the case $s = \iota(\Dmc)$, we may construct the set
  of all types and then remove all types $t$ such that one of the following
  conditions is satisfied:
  \begin{itemize}

  \item there is an $A(c) \in \Dmc$ with $A \notin t$;

  \item there is $r(c,c) \in \Dmc$ and $\forall r . C \in t$
    with $C \notin t$;

  \item there is $r(c,c) \in \Dmc$ and $\forall r^- . C \in t$
    with $C \notin t$.

  \end{itemize}
  We may then construct an IOA $\gamma \in \Theta(s)$ from each of the
  remaining types by setting, for $1 \leq i \leq k$,
  $\gamma^{\mn{in}}(i)=\{ C \mid \forall r . C \in t \}$ and 
  $\gamma^{\mn{out}}(i)=\{ \forall r C \mid \forall r . C \in t \}$.
  All
  this is clearly possible in time $O(k) \cdot 2^{O(|\Omc|)}$. In
  particular, the set of all types can be constructed by enumerating
  all subsets $t \subseteq\mn{sub}(\Omc)$ and then checking for
  satisfiability w.r.t.\ \Omc, which is possible in time
  $2^{O(|\Omc|)}$ using a type elimination procedure
  \cite{DBLP:books/daglib/0041477}.  The overall running time achieved
  is thus $2^{O(|\Omc| \cdot k)} \cdot |\Dmc|$, which is $2^{O(|Q| \cdot k)} \cdot |\Dmc|$.
\end{proof}

\subsection{\ALCI with UCQs}

We first argue that we may assume w.l.o.g.\ that $Q$ is
Boolean. Assume that it is not. Then we are given as an additional
input a tuple $\bar c \in \mn{adom}(\Dmc_0)^{|\bar x |}$.  Let
$\bar x = x_1,\dots,x_n$ and $\bar c=c_1,\dots,c_n$. We introduce
fresh concept names $A_1,\dots,A_n$. Let the Boolean UCQ $q'$ be
obtained from $q$ by adding the atoms $A_1(x_1),\dots,A_n(x_n)$ to
every CQ in it and quantifying all answer variables. Let the database
$\Dmc'_0$ be obtained from $\Dmc_0$ by adding the facts
$A_1(c_1),\dots,A_n(c_n)$. We may then decide whether
$\bar c \in Q(\Dmc_0)$ by checking whether $\Dmc'_0 \models Q'$ where
$Q'=(\Omc,q')$. It remains to note that $\Dmc'_0$ admits a
$k$-expression that is isomorphic to $s_0$ and can be constructed from
it in linear time.

\qhatworksalclem*
\noindent
\begin{proof}
  Let $\Imc$ be a tree-extended model of $\Dmc_0$ and $\Omc_q$
  such that $\Imc \models q$.  We construct a disjunct $\widehat{q}_d$ of
  $\widehat q$ such that there is a homomorphism from $\widehat{q}_d$ to
  $\Imc|_{\mn{adom}(\Dmc_0)}$.  Since $\Imc \models q$, there is a
  disjunct $q_d$ of $q$ and a homomorphism $h$ from
  $q_d$ to $\Imc$.  Consider the contraction $p$ of
  $q_d$ obtained by identifying variables $x$ and $y$ if
  $h(x)=h(y)$. Let the restriction of $h$ to the variables in $p$ also
  be called $h$. Then $h$ is an injective homomorphism from $p$ to
  $\Imc$.  Let $S = h^{-1}(\mn{adom}(\Dmc_0))$.  We define $\widehat{q}_d$ to be
  $p|_S$ extended as follows: For every maximal connected component
  $p'$ of $p^-$ that contains exactly one variable $x_0 \in S$, add
  $A_{p'(x_0)}(x_0)$ to $\widehat{q}_d$.  
  For every maximal connected component
  $p'$ of $p^-$ that contains no variable from~$S$, add $A_{p'}(z)$ to
  $\widehat{q}_d$, where $z$ is a new variable. 
  
  To prove $\Imc|_{\mn{adom}(\Dmc_0)} \models \widehat{q}_d$ we construct a homomorphism
  $\widehat{h}$ from $\widehat{q}_d$ to $\Imc|_{\mn{adom}(\Dmc_0)}$. For every $x \in S$ we define $\widehat{h}(x) = h(x)$ 
  and every newly introduced variable $z \in \mn{var}(\widehat{q}_d) \setminus S$ that was introduced for a maximal connected component $p'$ of $p^-$ is mapped to any constant of the database
  that lies in the same connected component as
  the images of $\mn{var}(p')$ under $h$.

  Let us first consider the atoms in $\widehat q_d$ of the form $A_{p'(x_0)}(x_0)$.
  Since $h$ is a homomorphism from $p$ to $\Imc$,
  we have $h(x_0) \in C_{p'(x_0)}^\Imc$.
  Consequently $\widehat h(x_0)$ satisfies $A_{p'(x_0)}(x_0)$ because of 
  the CI $C_{p'(x_0)} \sqsubseteq A_{p'(x_0)} \in \Omc_q$. 
  Now consider the atoms of the form $A_{p'}(z)$.
  Analogously to before there has to be an $a \in C_{p'}^\Imc$.
  Using the CIs $C_{p'} \sqsubseteq A_{p'}$ and $\exists r . A_{p'} \sqsubseteq A_{p'}$ and tree extendedness, we can argue 
  that $b \in A_{p'}^{\Imc|_{\mn{adom}(\Dmc_0)}}$ for every $b$ that lies in the same connected component as $a$.
  Thus $\Imc|_{\mn{adom}(\Dmc_0)} \models \widehat{q}_d$ and by definition, $\widehat{q}_d$ is a
  disjunct of $\widehat q$, which implies
  $\Imc|_{\mn{adom}(\Dmc_0)} \models \widehat q$.
  
  For the other direction, assume that $\Imc|_{\mn{adom}(\Dmc_0)} \models \widehat q$.
  Thus, there is a disjunct $\widehat{q}_d$ of $\widehat q$ and a homomorphism $\widehat{h}$ from $\widehat{q}_d$
  to $\Imc|_{\mn{adom}(\Dmc_0)}$.
  Let $p$ be the disjunct of $q$ that $\widehat{q}_d$ was constructed from.
  We want to show $\Imc \models p$ so we construct a homomorphism $h$ from $p$ to $\Imc$.
  The following observation is a straightforward implication of the definition of $C_p$ 
  and helps us construct partial homomorphisms.
  \begin{enumerate}[label=$(*)$]
    \item Let $\Imc$ be a tree extended model of $\Omc_q$ and $p(x)$ a unary tree CQ. 
    For any $a \in C_{p(x)}^\Imc$ we can construct a homomorphism $\bar h$ from $p(x)$ to $\Imc$ with $\bar h(x) = a$. 
  \end{enumerate}
  We define $h(x) = \widehat{h}(x)$ for all $x \in \mn{var}(p) \cap \mn{var}(\widehat{q}_d)$.
  According to the construction of $\widehat q$ we can categorize the remaining variables in $p$ 
  as being in one of two types of tree CQs.
  We construct a homomorphism for each of those subqueries.
  \begin{enumerate}
    \item A tree CQ $p'$ that is a subquery of $p$ with $\mn{var}(p') \cap \mn{var}(\widehat{q}_d) = \{x_0\}$. 
    In this case there is single shared variable $x_0$ with $A_{p'(x_0)}(x_0) \in \widehat{q}_d$ as defined in step~1 
    of the construction of $\widehat q$.
    We argue that there is a homomorphism $\bar{h}$ from $p'$ to $\Imc$.
    First, we know $\widehat h(x_0) \in C_{p'(x_0)}^\Imc$ because of the CI $A_{p'(x_0)} \sqsubseteq C_{p'(x_0)}$.
    Thus, $(*)$ applied on $p'(x_0)$ and $\widehat h(x_0)$ gives us the homomorphism $\bar{h}$ with $\bar{h}(x_0) = \widehat{h}(x_0)$.
    \item A tree CQ $p'$ that is a subquery of $p$ with $\mn{var}(p') \cap \mn{var}(\widehat{q}_d) = \emptyset$. 
    According to step~2 in the construction of $\widehat q$ there is an $A_{p'}(z) \in \widehat{q}_d$ with $z$ a fresh variable.
    The CI $A_{p'} \sqsubseteq \exists r_{p'} . C_{p'}$ implies there is a domain element 
    $a \in \Delta^\Imc$ with $a \in C_{p'}^\Imc$.
    Let $z'$ be the variable that was chosen as the root
    to interpret $C_{p'}$ as a $\ALCI$ concept.
    Then $(*)$ applied on $p'(z)$ and $a$ and gives us a homomorphism $\bar{h}$ from $p'(z')$ to $\Imc$ 
    and thus also from $p'$ to $\Imc$.
    Note that $\bar{h}$ maps only to variables not used in any other partial homomorphism 
    because for each subquery $p'$ there is a unique CI using the fresh role $r_{p'}$.
  \end{enumerate}

  It is easy to show that all these partial homomorphisms
  map shared variables to the same domain element.
  Thus combining them gives us the required homomorphism from $p$ to $\Imc$ and $\Imc \models p$ 
  implies $\Imc \models q$.
\end{proof}

We add a technical detail left out in the main part of the paper for
readability. We assume that the ontology $\Omc_q$ contains an
inclusion $\top \sqsubseteq A_\top$ and $A_\top$ does not occur
elsewhere, that is, $A_\top$ is a concept name that represents
$\top$. We may then also assume that every CQ $p$ in $\widehat q$
contains the atom $A_\top(x)$ for every $x \in \mn{var}(p)$. This is
important for the case $s = \iota(\Dmc)$ in the proof below.

\alciucqcorrectnesslem*
\noindent
\begin{proof}
We prove the statement by induction on the structure of $s$.
If $s = \iota(\Dmc)$, the statement holds by definition.
For the induction step, there are three cases to consider.

\smallskip
\noindent
\textbf{Case 1:} $s = s_1 \oplus s_2$. First, let $\gamma \in \Theta(s)$, we need to show that
there is a tree-extended model $\Imc$ of $\Dmc_s$ and $\Omc_q$
such that $\gamma = \gamma_{\Imc, \Dmc_s}$.
By definition of $\Theta$, there is a pair of decorated IOAs
$(\gamma_1, \gamma_2) \in \Theta(s_1) \times \Theta(s_2)$ that produced $\gamma$.
By the induction hypothesis, there are tree-extended models
$\Imc_j$ of $\Dmc_{s_j}$ and $\Omc_q$, $j \in \{1, 2\}$
such that $\gamma_j = \gamma_{\Imc_j, \Dmc_{s_j}}$.
Let $\Imc$ be the disjoint union of $\Imc_1$ and $\Imc_2$.
Clearly, $\Imc$ is a tree-extended model of $\Dmc_s$ and $\Omc_q$.

We argue that $\gamma = \gamma_{\Imc, \Dmc_s}$.
Since the extensions of concepts in $\Imc_1$ and $\Imc_2$ are not changed by
taking the disjoint union of these two models,
we have
\begin{align*}
\gamma^\mn{in} &= \gamma_1^\mn{in} \cap \gamma_2^\mn{in}\\
&= \gamma_{\Imc_1, \Dmc_{s_1}}^\mn{in} \cap \gamma_{\Imc_2, \Dmc_{s_2}}^\mn{in}\\
&= \mn{cl}^\ast(\Omc) \cap \bigcap_{L_i(c) \in \Dmc_{s_1}} \mn{tp}_{\Imc_1}(c) \cap \bigcap_{L_i(c) \in \Dmc_{s_2}} \mn{tp}_{\Imc_2}(c)\\
&= \mn{cl}^\ast(\Omc) \cap \bigcap_{L_i(c) \in \Dmc_{s}} \mn{tp}_{\Imc}(c)\\
&= \gamma^\mn{in}_{\Imc, \Dmc_s}
\end{align*}
and by a similar calculation, $\gamma^\mn{out} = \gamma_{\Imc, \Dmc_s}^\mn{out}$.

It remains to argue that $\gamma^S = \gamma_{\Imc, \Dmc_s}^S$.
To show $\gamma^S \subseteq \gamma_{\Imc, \Dmc_s}^S$, let $(p, f) \in \gamma^S$.
If $(p, f) \in \gamma_1^S \cup \gamma_2^S$, this is witnessed by a homomorphism
$h$ from $p$ to $\Imc_j$ for some $j \in \{1, 2\}$. The same $h$ is a homomorphism from
$p$ to $\Imc$ that witnesses $(p, f) \in \gamma_{\Imc, \Dmc_s}^S$. Otherwise, $(p, f)$ takes the form
$(p_1 \cup p_2, f_1 \cup f_2)$ with $(p_i,f_i) \in \gamma^S_i$ for $i \in \{1,2\}$,
$p_1 \cup p_2 \in \Qmc$ and $\mn{var}(p_1) \cap \mn{var}(p_2) = \emptyset$.
Thus, there are homomorphisms $h_i$ from $p_i$ to $\Imc_i$ that witness
$(p_i, f_i) \in \gamma_i^S$ for $i \in \{1, 2\}$.
The union of these two homomorphisms witnesses that $(p, f) \in \gamma_{\Imc, \Dmc_s}^S$.
To show $\gamma^S \supseteq \gamma_{\Imc, \Dmc_s}^S$, let $(p, f) \in \gamma_{\Imc, \Dmc_s}^S$,
and let $h$ be a homomorphism from $p$ to $\Imc$ witnessing this.
Since $\Imc$ is the disjoint union of $\Imc_1$ and $\Imc_2$, we can
distinguish two cases:
First, if $\mn{range}(h) \subseteq \Delta^{\Imc_i}$ for some $i \in \{1, 2\}$,
then $h$ witnesses that $(p, f) \in \gamma_i^S \subseteq \gamma^S$.
Otherwise, we can represent $p$ as the disjoint union of $p_1$ and $p_2$,
such that $h|_{p_1}$ is a homomorphism from $p_1$ to $\Imc_1$ and
$h|_{p_2}$ is a homomorphism from $p_2$ to $\Imc_2$, and we set
$f_i = f|_{\mn{var}(p_i)}$ for $i \in \{1, 2\}$.
These two homomorphisms witness that $(p_i, f_i) \in \gamma_{\Imc_i, \Dmc_{s_i}}^S$
for $i \in \{1, 2\}$. By the induction hypothesis, we have $(p_i, f_i) \in \gamma_i^S$ for $i \in \{1, 2\}$,
and it follows that $(p, f) = (p_1 \cup p_2, f_1 \cup f_2) \in \gamma^S$.
This concludes the argument that $\gamma^S = \gamma_{\Imc, \Dmc_s}^S$,
and thus, $\gamma = \gamma_{\Imc, \Dmc_s}$.

For the other direction, let $\Imc$ be a tree-extended model of $\Dmc_s$ and $\Omc$ and consider
$\gamma = \gamma_{\Imc, \Dmc_s}$. We need to show that $\gamma \in \Theta(s)$.
Since $\Imc$ is tree-extended and $\Dmc_s$ is the disjoint union of $\Dmc_{s_1}$ and $\Dmc_{s_2}$,
there are interpretations $\Imc_1$ and $\Imc_2$ such that
$\Imc$ is the disjoint union of $\Imc_1$ and $\Imc_2$, and
$\Imc_i$ is a tree-extended model of $\Dmc_{s_i}$ and $\Omc$ for $i \in \{1, 2\}$.
Thus, there are decorated IOAs $\gamma_i = \gamma_{\Imc_i, \Dmc_{s_i}}$ for $i \in \{1, 2\}$.
By the induction hypothesis, we have $\gamma_i \in \Theta(s_i)$ for $i \in \{1, 2\}$.
It then follows that $\gamma$ is the decorated IOA obtained from the
pair $(\gamma_1, \gamma_2)$. The argument is analogous to the proof that
$\gamma = \gamma_{\Imc, \Dmc_s}$ from the previous direction.
Thus, $\gamma \in \Theta(s)$.

\smallskip
\noindent
\textbf{Case 2:} $s = \alpha_{i, j}^r(s')$. First, let $\gamma \in \Theta(s)$.
We need to show that there is a tree-extended model $\Imc$ of $\Dmc_s$ and $\Omc$
such that $\gamma = \gamma_{\Imc, \Dmc_s}$. Let $\widehat{\gamma} \in \Theta(s')$ be
the decorated IOA that $\gamma$ was obtained from.
By the induction hypothesis, there is a tree-extended model $\Imc'$ of $\Dmc_{s'}$ and $\Omc$
such that $\widehat{\gamma} = \gamma_{\Imc', \Dmc_{s'}}$.
Let $\Imc = \Imc' \cup \{r(d,e) \mid L_i(d) \in \Dmc_{s'} \text{ and } L_j(e) \in \Dmc_{s'}\}$.
Then $\Imc$ is a model of $\Dmc_s$.
To show that $\Imc$ is also a model of $\Omc_q$, it suffices to show that
$\mn{tp}_\Imc(c)= \mn{tp}_{\Imc'}(c)$ for all $c \in \mn{adom}(\Dmc_{s})$.
This can be shown by a straightforward induction on the structure of concepts
$C \in \mn{cl}^\ast(\Omc_q)$: $C \in \mn{tp}_\Imc(c)$ iff $C \in \mn{tp}_{\Imc'}(c)$
for all $c \in \mn{adom}(\Dmc_{s})$.
The only interesting case in the induction is the case where $C$ takes the form
$\forall r. C'$, but this case is proved using Conditions~a) and~b).
Thus, $\Imc$ is also a model of $\Omc_q$ and $\gamma_{\Imc,\Dmc_s} = \gamma_{\Imc', \Dmc_{s'}} = \gamma$ as required.
We argue that $\gamma = \gamma_{\Imc, \Dmc_s}$.
It is easy to see that $\gamma^\mn{in} = \gamma_{\Imc, \Dmc_s}^\mn{in}$
and $\gamma^\mn{out} = \gamma_{\Imc, \Dmc_s}^\mn{out}$.
To show that $\gamma^S \subseteq \gamma_{\Imc, \Dmc_s}^S$, let $(p, f) \in \gamma^S$.
If also $(p, f) \in \widehat{\gamma}^S$, then $(p, f) \in \gamma_{\Imc, \Dmc_s}^S$ follows using
the induction hypothesis. If $(p, f) \notin \widehat{\gamma}^S$, then $(p, f)$ is one of the pairs that was
added to $\gamma^S$ when going from $\Theta(s')$ to $\Theta(s)$.
Thus, there is a pair $(p', f) \in \widehat{\gamma}^S$ and $p$ can be obtained from
$p'$ by adding some atoms $r(x,y)$ subject to the condition that
$f(x)=i$ and $f(y)=j$.
Since $\widehat{\gamma} = \gamma_{\Imc', \Dmc_s'}$, there is a homomorphism
$h$ from $p'$ to $\Imc'$ that witnesses $(p', f) \in \gamma_{\Imc', \Dmc_s'}$.
This $h$ is also a homomorphism from $p$ to $\Imc$
that witnesses $(p, f) \in \gamma_{\Imc, \Dmc_s}^S$, since
$\Imc$ was constructed such that $(h(x), h(y)) \in r^\Imc$ for every
pair $(x, y)$ with $f(x) = i$ and $f(y) = j$.
To show $\gamma^S \supseteq \gamma_{\Imc, \Dmc_s}^S$, let $(p, f) \in \gamma_{\Imc, \Dmc_s}^S$.
There is a homomorphism $h$ from $p$ to $\Imc$ witnessing this.
Let $p'$ be obtained from $p$ by removing all $r$-atoms $r(x,y)$ where
$L_i(x) \in \Dmc_s$ and $L_j(y) \in \Dmc_s$.
The homomorphism $h$ also witnesses that $(p', f) \in \gamma_{\Imc', \Dmc_{s'}}^S = \widehat{\gamma}^S$.
By definition of $\Theta$, $(p', f) \in \widehat{\gamma}^S$ yields $(p, f) \in \gamma^S$.

For the other direction, let $\Imc$ be a tree-extended model of $\Dmc_s$ and $\Omc_q$,
and consider $\gamma = \gamma_{\Imc, \Dmc_s}$.
We need to show that $\gamma \in \Theta(s)$.
Let $\Imc'$ be the interpretation obtained from $\Imc$ as follows:
Removing all pairs $(d, e)$ from
$r^\Imc$ where $L_i(d) \in \Dmc$, $L_j(e) \in \Dmc$ and $r(d, e) \notin \Dmc_{s'}$. For each such removed pair $(d, e)$, introduce two fresh
individuals $d'$ and $e'$ to $\Delta^{\Imc'}$,
add $(d, e')$ and $(d', e)$ to $r^{\Imc'}$,
and attach a tree model of $\mn{tp}_{\Imc}(d)$ to $d'$
and a tree model of $\mn{tp}_{\Imc}(e)$ to $e'$.
Clearly, $\Imc'$ is a tree-extended model of $\Dmc_{s'}$.
To show that $\Imc'$ is also a model of $\Omc_q$, it suffices to show that
$\mn{tp}_\Imc(c)= \mn{tp}_{\Imc'}(c)$ for all $c \in \mn{adom}(\Dmc_{s})$.
This can be shown by a straightforward induction on the structure of concepts
$C \in \mn{cl}^\ast(\Omc_q)$: $C \in \mn{tp}_\Imc(c)$ iff $C \in \mn{tp}_{\Imc'}(c)$
for all $c \in \mn{adom}(\Dmc_{s})$.
The only interesting case in the induction is the case where $C$ takes the form
$\forall r. C'$, but this case is proved using Conditions~a) and~b).
Thus, $\Imc$ is also a model of $\Omc_q$ and
$\gamma_{\Imc,\Dmc_s} = \gamma_{\Imc', \Dmc_{s'}} = \gamma$ as required.

By the induction hypothesis, there is a decorated IOA
$\widehat{\gamma} \in \Theta(s')$
such that $\widehat{\gamma} = \gamma_{\Imc', \Dmc_{s'}}$.
Now it can be argued analogously to the previous direction,
that $\gamma$ is the decorated IOA obtained from $\widehat{\gamma}$,
so $\gamma \in \Theta(s)$.

\smallskip
\noindent
\textbf{Case 3:} $s = \rho_{i \rightarrow j}(s')$. First, let $\gamma \in \Theta(s)$.
We need to show that there is a tree-extended model $\Imc$ of $\Dmc_s$ and $\Omc$
such that $\gamma = \gamma_{\Imc, \Dmc_s}$.
Let $\widehat{\gamma} \in \Theta(s')$ be the decorated IOA that $\gamma$ was obtained from.
By the induction hypothesis, there is a tree-extended model $\Imc'$ of $\Dmc_{s'}$ and $\Omc$ such that
$\widehat{\gamma} = \gamma_{\Imc', \Dmc_{s'}}$. Let $\Imc$ be the interpretation obtained from $\Imc'$
by setting $L_i^\Imc = \emptyset$ and $L_j^\Imc = L_i^{\Imc'} \cup L_j^{\Imc'}$.
Since $\Imc'$ is a tree-extended model of $\Dmc_{s'}$, $\Imc$ is a tree-extended
model of $\Dmc_s$. It is easy to show that $\gamma = \gamma_{\Imc, \Dmc_s}$.

For the other direction, let $\Imc$ be a tree-extended model of $\Dmc_s$ and $\Omc$,
and set $\gamma = \gamma_{\Imc, \Dmc_s}$. We need to show that $\gamma \in \Theta(s)$.
Let $\Imc'$ be the interpretation that is obtained from $\Imc$,
but $L_i^{\Imc'} = \{a \in \Delta^{\Imc'} \mid L_i(a) \in \Dmc_{s'}\}$
and $L_j^{\Imc'} = \{a \in \Delta^{\Imc'} \mid L_j(a) \in \Dmc_{s'}\}$.
This interpretation $\Imc'$ is a tree-extended model of $\Dmc_{s'}$ and $\Omc$, so
by the induction hypothesis, there is a decorated IOA $\widehat{\gamma} \in \Theta(s')$ with
$\widehat{\gamma} = \gamma_{\Imc', \Dmc_{s'}}$. It is straightforward to show that
$\gamma$ is the decorated IOA obtained from $\widehat{\gamma}$, so $\gamma \in \Theta(s)$.
\end{proof}

\alciucqtwoexpthm*
\noindent
\begin{proof}
First, we prove correctness of the algorithm, that is,
$\Dmc_0, \Omc \not \models q$ if and only if there is a $\gamma \in \Theta(s)$
such that $\gamma^S$ contains no pair $(p, f)$ with $p$ a CQ in $\widehat q$. 
First, let $\Dmc_0, \Omc \not \models q$.
Then, by Lemma~\ref{lem:treeextended1},
there is a tree-extended model $\Imc$ of $\Dmc_0$ and $\Omc$
such that $\Imc \not \models q$.
The model $\Imc$ of $\Omc$ can be extended to a model $\Imc'$ of $\Omc_q$
as follows:
\begin{itemize}
\item For every $p \in \mn{trees}(q)$ and $x \in \mn{var}(p)$, we have the CIs
$A_{p(x)} \sqsubseteq C_{p(x)}$ and $C_{p(x)} \sqsubseteq A_{p(x)}$ in $\Omc_q$.
To satisfy these, for every $d \in C_{p(x)}^\Imc$, add $d$ to $A_{p(x)}^{\Imc'}$.
\item For every $p \in \mn{trees}(q)$,
we have the CIs $A_p \sqsubseteq \exists r_p. C_p$, $C_p \sqsubseteq A_p$
and $\exists r.A_p \sqsubseteq A_p$ in $\Omc_q$ for every role $r$ used in $\Omc$.
To satisfy these, repeat the following process indefinitely: For every $d \in C_p^{\Imc'}$,
add every $e \in \Delta^{\Imc'}$ that lies in the same connected component as $d$ to $A_p^{\Imc'}$,
and, if $e$ does not yet have an $r_p$-successor,
introduce a fresh $r_p$-successor of $e$ to an element
that is the root of a tree model of $C_p$.
Finally, let $\Imc'$ be the limit of this process.
\end{itemize}
It can be checked that $\Imc'$ is a tree-extended model of $\Omc_q$ and $\Imc \not \models q$.
By Lemma~\ref{lem:qhatworksalc}, we have
$\Imc'|_{\mn{adom}(\Dmc_0)} \not \models \widehat q$.
Consider the decorated IOA $\gamma = \gamma_{\Imc', \Dmc_s}$.
Since $\Imc'|_{\mn{adom}(\Dmc_0)} \not \models \widehat q$, $\gamma^S$ does not contain
any pair $(p, f)$ with $p$ a CQ in $\widehat q$.
By Lemma~\ref{lem:alci-ucq-correctness}, we have $\gamma \in \Theta(s)$.

For the other direction, let $\gamma \in \Theta(s)$ such that $\gamma^S$ contains no pair
$(p, f)$ with $p$ a CQ in $\widehat q$. By Lemma~\ref{lem:alci-ucq-correctness},
$\gamma = \gamma_{\Imc, \Dmc_s}$
for some tree-extended model $\Imc$ of $\Dmc_s$ and $\Omc_q$.
It remains to show that $\Imc \not \models q$.
Assume for the sake of contradiction that $\Imc \models q$.
By Lemma~\ref{lem:qhatworksalc}, $\Imc|_{\mn{adom}(\Dmc_0)} \models \widehat q$,
so there is a homomorphism from some disjunct $p$ of $q$ to $\Imc|_{\mn{adom}(\Dmc_0)}$.
This implies that $\gamma^S$ contains the pair $(p, f)$, where $f(x)$ is the unique label $i$ such that
$L_i(h(x)) \in \Dmc_s$ for every $x \in \mn{var}(p)$. This is a contradiction, since we assumed
that $\gamma^S$ contains no such pair. Thus, the algorithm is correct.

Now we analyse the running time.
We first analyse the number of CQs in $\Qmc$.
For any disjunct $p$ of $q$, there are at most $|p|^{|p|}$ many contractions of $p$ and
at most $2^{|p|}$ ways to choose a set $S$.
Once a set $S$ is chosen, this yields a unique CQ in $\widehat q$,
so $\widehat q$ has at most
$|q| \cdot |q|^{|q|} \cdot 2^{|q|} = 2^{O(|q| \log(|q|))}$
disjuncts in $\widehat q$,
where each disjunct is of size at most $|q|$. The set $\Qmc$ consists of all subqueries
of disjuncts of $\widehat q$, so there are at most
$2^{|q|} \cdot 2^{O(|q| \log(|q|))} = 2^{O(|q| \log(|q|))}$ CQs in $\Qmc$.

The number of decorated IOAs can be bounded as follows.
For $\gamma^\mn{in}$, there are at most $2^{|\Omc|\cdot k}$ possibilities,
the same bound holds for the number of possible $\gamma^\mn{out}$.
For $\gamma^S$, we first count the number of pairs $(p, f)$.
As mentioned above, there are $2^{O(|q| \log(|q|))}$ many possibilities for $p$.
For $f$, there are $k^{|q|}$ many possibilities.
This yields at most $2^{O(|q| \log(|q|))} \cdot k^{|q|}$ many different pairs $(p, f)$,
and at most
$2^{k^{O(|q| \log(|q|))}}$ 
  many possible sets $\gamma^S$.  Combining the possibilities for
  $\gamma^\mn{in}$, $\gamma^\mn{out}$ and $\gamma^S$, we have at most
  $2^{|\Omc| \cdot  k^{O(|q| \log(|q|))}}$
  different
  decorated IOAs.

The algorithm goes bottom up through all subexpressions of the given $k$-expression,
which takes $O(|\Dmc_0|)$ iterations. In each iteration,
a set of decorated IOAs is computed, of which there are
at most
$2^{|\Omc|\cdot 
  k^{O(|q| \log(|q|))}}$ many.
Constructing the set $\Theta(s)$, where $s=s_1 \oplus s_2$ or $s = \alpha_{i,j}^{r}(s')$
or $s = \rho_{i \rightarrow j} (s')$ is straightforward and takes
time $2^{|\Omc|\cdot k^{O(|q| \log(|q|))}}$.
If $s = \iota(\Dmc)$ and $L_i(c)$ is the unique fact of this form in
\Dmc, we construct $\Theta(s)$ as follows:
Iterate over all $\Omc_q$ types $t$ such that  $\mn{tp}_{\Imc}(c)=t$ for some model \Imc of
\Dmc and $\Omc_q$.
For each such $t$, we compute one decorated IOA $\gamma$, where
\begin{align*}
\gamma^{\mn{in}}(i) &= \mn{cl}^\ast(\Omc_q) \cap t  \text{ and } \gamma^{\mn{in}}(j) = \emptyset \text{ for all }j \neq i, \\
\gamma^{\mn{out}}(i) &= \mn{cl}^\forall(\Omc_q) \cap t  \text{ and } \gamma^{\mn{in}}(j) = \emptyset \text{ for all }j \neq i.
\end{align*}
To compute $\gamma^S$, we first compute $\Imc|_{\{a\}}$ for some tree-extended
model $\Imc$ of $\Dmc$ and $\Omc_q$ that satisfies $\mn{tp}_\Imc(a) = t$.
Iterate over all $p \in \Qmc$ and check whether the map that sends
every variable of $p$ to $a$ is a homomorphism from $p$
to $\Imc|_{\{a\}}$.
If this is the case, add the pair $(p, f)$ to $\gamma^S$, where $f(x) = i$
for every $x \in \mn{var}(p)$.
This can be done in time $2^{|\Omc| \cdot 2^{O(|q| \cdot \log |q|)}}$.
Overall, the algorithm has running time
$2^{|\Omc|\cdot 
  k^{O(|q| \log(|q|))}} \cdot |\Dmc|$. 
\end{proof}

\subsection{\GFtwo with AQs}

\gftwoaqalgocorrectnesslem*
\noindent
\begin{proof}
  We prove by induction on the structure of $s$ that the equality in Lemma \ref{lem:1exp_correctnes_gf2_aq} holds.
  For the base case $s = \iota(\Dmc)$ it holds per definition.
  There are three cases in the induction step that we further split into the ``$\subseteq$'' and ``$\supseteq$'' direction.

  $s = s_1 \oplus s_2$, ``$\subseteq$'': Let $\gamma \in \Theta(s)$ be
  a model abstraction. We have to identify a model \Imc of $\Dmc_s$
  and \Omc such that $\gamma=\gamma_{\Imc,\Dmc_s}$.  Since
  $\gamma \in \Theta(s)$, there are $\gamma_1 \in \Theta(s_1)$ and
  $\gamma_2 \in \Theta(s_2)$ that result in $\gamma$ being included in
  $\Theta(s)$.  By the induction hypothesis, we find models $\Imc_i$ of
  $\Dmc_{s_i}$ and \Omc such that
  $\gamma_i = \gamma_{\Imc_i,\Dmc_{s_i}}$ for $i \in \{1,2\}$.  We can
  assume the domains of $\Imc_1$ and $\Imc_2$ to be disjoint because
  the active domains of $\Dmc_{s_1}$ and $\Dmc_{s_2}$ are disjoint.
  Choose as \Imc the union of $\Imc_1$ and $\Imc_2$, which is a model
  of $\Dmc_s$ and \Omc.  The construction of $\gamma \in \Theta(s)$
  requires $\gamma_1^0 = \gamma_2^0$ so together with the disjoint
  domains this gives us $\mn{tp}^1_\Imc(c)= \mn{tp}^1_{\Imc_i}(c)$ for
  all $c \in \mn{adom}(\Dmc_{s_i})$, $i \in \{1,2\}$.  This implies
  $\gamma^T = \gamma_{\Imc, \Dmc_s}^T$.  Note that if there are two
  constants $c,c' \in \mn{adom}(\Dmc)$ in a database $\Dmc$ with no
  role link between them, then the multi-edge $M_\Dmc(c,c')$ is 
  empty.  Thus also $\gamma^E = \gamma_{\Imc, \Dmc_s}^E$ and therefore
  $\gamma = \gamma_{\Imc,\Dmc_s}$, as required.

  $s = s_1 \oplus s_2$, ``$\supseteq$'': Let $\Imc$ be a model of $\Dmc_s$
  and $\Omc$. Then $\Imc$ is also a model of $\Dmc_{s_1}$ and
  $\Dmc_{s_2}$. Thus by the induction hypothesis,
  $\gamma_{\Imc,\Dmc_{s_i}} \in \Theta(s_i)$ for $i \in \{1,2\}$.
  What remains to be shown is that this results in $\gamma_{\Imc, \Dmc_{s}}$
  to be included in $\Theta(s)$. The argument is  similar to the one
  used in the
  converse direction.

  $s = \alpha_{i, j}^r(s')$, ``$\subseteq$'': Let
  $\gamma \in \Theta(s)$ be a model abstraction.  We have to identify
  a model \Imc of $\Dmc_s$ and \Omc such that
  $\gamma=\gamma_{\Imc,\Dmc_s}$.  Since $\gamma \in \Theta(s)$, there
  is a $\widehat{\gamma} \in \Theta(s')$ that results in $\gamma$ being
  included in $\Theta(s)$. Note that $\widehat{\gamma}$ satisfies
  Condition~$(*)$ from the `$s = \alpha_{i, j}^r(s')$' case of the
  construction of $\Theta(s)$ in the
  main body of the paper.  By the induction hypothesis, we find a model
  $\Imc'$ of $\Dmc_{s'}$ and $\Omc$ such that
  $\gamma_{\Imc', \Dmc_{s'}} = \widehat{\gamma}$.  Let interpretation $\Imc'$
  be defined like $\Omc$, except that $$r^{\Imc'}=r^\Imc \cup 
\{r(c,c') \mid L_i(c) \in \Dmc_{s'} \text{ and }
  L_j(c') \in \Dmc_{s'}\}.$$  Then $\Imc$ is a model of $\Dmc_s$.  What
  remains to be shown is that $\Imc$ is also a model of $\Omc$.  Let
  $c,c' \in \mn{adom}(\Dmc_{s'})$ be any two constants such that
  $L_i(x) \in \mn{tp}^1_{\Imc'}(c)$ and
  $L_j(x) \in \mn{tp}^1_{\Imc'}(c')$.  Then the fact that $\widehat{\gamma}$
  satisfies $(*)$ gives us
  $\mn{tp}^1_{\Imc'}(c) \rightsquigarrow_t \mn{tp}^1_{\Imc'}(c')$ for
  some $t \in \mn{TP}_2$ with
  $$\{r(x,y) \mid r \in \NR \text{ and } r(c,c') \in \Dmc_{s}\}
  \subseteq t$$.  Note that the atomic formula $r(x,y)$ for the fact
  $r(c,c')$ that we just added to $\Dmc_{s'}$ is also part of $t$.  This
  implies $\mn{tp}^1_{\Imc}(c) \rightsquigarrow_t \mn{tp}^1_{\Imc}(c')$
  and therefore $\Imc$ is also a model of $\Omc$.  Thus, the
  definition of the model abstraction $\gamma_{\Imc,\Dmc_s}$ and the
  construction of $\gamma$ from $\widehat{\gamma}$ imply that
  $\gamma_{\Imc,\Dmc_s}=\gamma$, as required.

  $s = \alpha_{i, j}^r(s')$, ``$\supseteq$'': Let $\Imc$ be a model of
  $\Dmc_s$ and $\Omc$.  Then $\Imc$ is also a model of $\Dmc_{s'}$.
  By the induction hypothesis, $\gamma_{\Imc, \Dmc_{s'}} \in \Theta(s')$.
  The definition of the model abstraction $\gamma_{\Imc, \Dmc_{s'}}$
  and construction of $\Theta(s)$ imply that
  $\gamma_{\Imc, \Dmc_{s'}}$ results in $\gamma_{\Imc, \Dmc_{s}}$
  being included in $\Theta(s)$, as long as $\gamma_{\Imc, \Dmc_{s'}}$
  satisfies $(*)$.  This last point, of $\gamma_{\Imc, \Dmc_{s'}}$
  satisfying $(*)$, follows from $\Imc$ being a model of $\Dmc_{s}$.
  Thus, $\gamma_{\Imc, \Dmc_{s}} \in \Theta(s)$ as required.

  $s = \rho_{i \rightarrow j}(s')$, ``$\subseteq$'':
  Let $\gamma \in \Theta(s)$ be a model abstraction.
  We have to identify a model \Imc of $\Dmc_s$ and \Omc such that $\gamma=\gamma_{\Imc,\Dmc_s}$.
  Since $\gamma \in \Theta(s)$ there is a $\widehat{\gamma} \in \Theta(s')$ that results in $\gamma$ being included in $\Theta(s)$.
  By the induction hypothesis, we find a model $\Imc'$ of $\Dmc_{s'}$ and $\Omc$ such that $\gamma_{\Imc', \Dmc_{s'}} = \widehat{\gamma}$.
  Let $\Imc$ be just like $\Imc'$ except $L_i^\Imc = \emptyset$ and $L_j^\Imc = L_i^{\Imc'} \cup L_j^{\Imc'}$.
  Then $\Imc$ is a model of $\Dmc_s$ and $\Omc$ and the definition of the model abstraction $\gamma_{\Imc,\Dmc_s}$ and the construction of $\gamma$ from $\widehat{\gamma}$ imply that $\gamma_{\Imc,\Dmc_s} = \gamma$, as required.

  $s = \rho_{i \rightarrow j}(s')$, ``$\supseteq$'':
  Let $\Imc$ be a model of $\Dmc_s$ and $\Omc$.
  Moreover let $\Imc'$ be just like $\Imc$ except $L_i^{\Imc'} = \{a \mid L_i(a) \in \Dmc_{s'}\}$ and $L_j^{\Imc'} = L_j^{\Imc} \setminus L_i^{\Imc'}$.
  Then $\Imc'$ is a model of $\Dmc_{s'}$ and $\Omc$.
  Thus by the induction hypothesis, $\gamma_{\Imc', \Dmc_{s'}} \in \Theta(s')$.
  What remains to be shown is that this results in $\gamma_{\Imc, \Dmc_{s}}$ to be included in $\Theta(s)$. 
 We may argue similarly to the previous case.
\end{proof}

\gftwoaqtwoexpthm*
\noindent
\begin{proof}
  Correctness of the given algorithm follows immediately from Lemma
  \ref{lem:1exp_correctnes_gf2_aq}.  Next we analyze the running time.

  We start with noting that the sets of all 1-types $\mn{TP}_1$ and of
  all 2-types $\mn{TP}_2$ can be computed by enumerating all
  candidates $t \subseteq \mn{cl}_i(\Omc)$, $i \in \{1,2\}$, and then
  checking whether $t$ is satisfiable w.r.t.\ \Omc. The latter can be
  done in \ExpTime \cite{DBLP:journals/jsyml/Gradel99}, which yields
  an upper bound of $2^{2^{O(|\Omc|)}}$, sufficient for our purposes.
  However, it also seems clear that the standard type elimination
  procedure for \ALC can be adapted to GF$_2$ in a straightforward
  way, eliminating 1-types and 2-types simultaneously. This gives a
  $2^{O(|\Omc|)}$ upper bound.

  For every subexpression $s$ of $s_0$, we construct a set of model
  abstractions $\Theta(s)$. The number of model abstractions, and thus
  the size of the sets $\Theta(s)$, is clearly bounded by
  $2^{2^{O(|\Omc|)} \cdot k^2}$. Moreover, each set $\Theta(s)$ can
  clearly be constructed in time $2^{2^{O(|\Omc|)} \cdot k^2}$.  Thus,
  the algorithm has running time
  $2^{2^{O(|\Omc|)} \cdot k^2} \cdot |\Dmc|$, which is $2^{2^{O(|Q|)} \cdot k^2} \cdot |\Dmc|$.
\end{proof}

\section{Proofs for Section~\ref{sect:lower}}

\listcolgftwolem*

\noindent 
\begin{proof}
Let $G=(V, E)$ be an undirected graph of cliquewidth~2 and $C_v \subseteq \mathbb{N}$ be a list of possible colors for every $v \in V$. We construct a $\text{GF}_2$-ontology $\Omc$, a database~$\Dmc$ of cliquewidth~2, an AQ $q(x)$, and an answer candidate $a \in \mn{adom}(\Dmc)$ such that $G$ and $C_v$ are a yes-instance for \textsc{List Coloring} if and only if $\Omc, \Dmc \not \models q(a)$.
Let $k$ be the number of distinct colors that occur, w.l.o.g.\ let $\bigcup_{v \in V} C_v = \{1, \ldots, k\}$. Define $\ell = \lceil \log k \rceil$ and $L = \{1, \ldots, 2^\ell\}$.

We first construct the database \Dmc. We use concept names~$A_i^j$,
$i \in \{1, \ldots, \ell\}$ and $j \in \{0, 1\}$, another concept
name~$B$, and two role names $r$ and $s$.  We further use, for every
$v \in V$ and $c \in L \setminus C_v$, constants $a_v$ and $a_v^c$.
Define
\begin{align*}
\Dmc =\ & \{r(x, y), r(y, x) \mid x \in \{a_v, a_v^{c_1}\}, y \in \{a_u, a_u^{c_2}\},\\
& \qquad \{u,v\} \in E, c_1 \in L \setminus C_v, c_2 \in L \setminus C_u\}\, \cup\\
& \{s(a_v, a_v^c) \mid v \in V, c \in L \setminus C_v\}\,\cup\\
& \{A_i^j(a_v^c) \mid v \in V, c \in L \setminus C_v \text{ and the } i \text{-th bit of } c \text{ is } j\}\,\cup\\
& \{B(a_v) \mid v \in V\}.
\end{align*}
This construction is similar to the reduction that achieves cliquewidth~3 described above. The differences are that we use two different binary symbols, $r$ for edges that were present in the original graph and $s$ for edges between a node $a_v$ and a newly introduced node $a_v^c$, and that these newly introduced nodes are marked by a fresh concept name $B$. Also, the facts of the form $r(x,y)$ where at least one of $x$ and $y$ is of the form $a_v^c$ do not represent edges of $G$, but they are merely introduced to reduce the cliquewidth from~3 to~2. We call facts of this type \emph{unintended}. All other $r$-facts are called \emph{intended} and by construction of $\Dmc$, an $r$-fact between two constants $a$ and $b$ is intended if and only if both $B(a)$ and $B(b)$ are in $\Dmc$.

We argue that $\Dmc$ has cliquewidth 2. Since $G$ has cliquewidth~2, there is a 2-expression $\sigma$ for $G$. We manipulate $\sigma$, to obtain a 2-expression for $\Dmc$. Every edge insertion from label $i$ to label $j$ in $G$ is replaced by two operations $\alpha_{i, j}^r$ and $\alpha_{j, i}^r$.
Whenever a constant $a_v$ with label $i \in \{1, 2\}$ is introduced, we also introduce the constants $a_v^c$ with label $3-i$, take the disjoint union of all nodes, apply the operation $\alpha_{i, 3-i}^s$, and then apply $\rho_{3-i \rightarrow i}$ so that
the constant $a_v$ carries the same label as all its neighbours $a_v^c$. This is clearly a 2-expression and it can be verified that this expression generates $\Dmc$. In particular, since $a_v$ and $a_v^c$
carry the same label, every $r$-fact that is generated by the 2-expression and that affects $a_v$ will also affect every $a_v^c$ in the same way. This generates the facts $r(a_v, a_u^{c_2})$, $r(a_u^{c_2}, a_v)$, $r(a_v^{c_1}, a_u)$, $r(a_u, a_v^{c_1})$, $r(a_v^{c_1}, a_u^{c_2})$ and $r(a_u^{c_2}, a_v^{c_1})$ for every edge $\{u,v\} \in E$, every color $c_1 \in L \setminus C_v$ and every color $c_2 \in L \setminus C_u$. 

Next, we construct the ontology. The ontology has three purposes. First, to assign a color to every constant in the database, second, to detect a defect in the coloring, and third, to propagate a defect symbol $D$ to everywhere, once a defect has been detected. To assign a color to every constant, we introduce the sentence $\forall x A_i^0(x) \leftrightarrow \neg A_i^1(x)$. To detect a defect in the coloring, we check whether the colors of adjacent constants are equal, but we are only interested in facts that are not unintended. This is achieved by the following two sentences:
\begin{align*}
&\forall x \forall y \left( s(x,y) \wedge \bigwedge_{i=1}^\ell A_i^0(x) \leftrightarrow A_i^0(y) \right) \rightarrow D(x)\\
& \forall x \forall y \, (B(x) \wedge B(y) \wedge r(x,y)\\\
& \qquad \wedge \bigwedge_{i=1}^\ell A_i^0(x) \leftrightarrow A_i^0(y) ) \rightarrow D(x)
\end{align*}
The first sentence checks that the coloring is valid along $s$-facts. The second sentence checks that the coloring is valid along intended $r$-facts. Recall that an $r$-fact between two constants $a$ and $b$ is intended if and only if both $B(a)$ and $B(b)$ are in $\Dmc$. If a defect in the coloring is found, the concept name $D$ is derived.
Finally, we propagate the symbol $D$ to the whole connected component using the following sentence:
$$
\forall x\forall y \Big(\big(r(x,y) \lor s(x,y)\big) \land \big(D(x) \lor D(y)\big)\Big) \rightarrow D(x) \land D(y)
$$

We define the query to be $q(x) = D(x)$ and the answer candidate to be an arbitrary constant from $\mn{adom}(\Dmc)$.

It remains to argue that the reduction is correct. Assume that the input graph $G$ and the $C_v$
form a yes-instance for \textsc{List Coloring} and let $f : V \rightarrow \mathbb{N}$ be a coloring
witnessing this. We construct a model $\Imc$ of $\Omc$ and $\Dmc$ such that $\Imc \not \models D(a)$, where $\Imc$ is obtained from $\Dmc$ by adding, for every vertex $v \in V$ the facts $A_i^j(a_v)$ that represent the color $f(v)$. It is straightforward to check that $\Imc$ is a model of $\Omc$ and $\Dmc$ and that $\Imc \not \models q(a)$.

For the other direction, assume that $\Omc, \Dmc \not \models q(a)$ and let $\Imc$ be a model
witnessing this. Since $\Imc$ is a model of $\Omc$, every constant of the form $a_v$ is labelled
by a unique sequence of of facts $A_i^j(a_v)$ encoding a color $c_v$. Let $f(v) = c_v$ for every $v \in V$. We argue that $f$ is a coloring of $G$ with $c_v \in C_v$ for every $v \in V$. Since $G$ is connected, so is $\Dmc$. Since the ontology propagates the concept name $D$ along all binary predicates, and since $\Imc \not \models D(a)$, we known that $D(a_v) \notin \Imc$ and $D(a_v^c) \notin \Imc$ for all $v \in V$ and $c \in C_v$. Consider any $s$-fact $s(b,b') \in \Imc$ or any intended $r$-fact $r(b, b') \in \Imc$. If $b$ and $b'$ had the same color, the ontology would imply that $D(b) \in \Imc$. But since $D(b) \notin \Imc$, the colors of $b$ and $b'$ are different. For every color $c \in L \setminus C_v$, we have $s(a_v, a_v^c) \in \Imc$. Since $a_v$ and $a_v^c$ have a different color
in $\Imc$, we have $f(v) \in C_v$. Thus, $f$ witnesses that $G$ is a yes-instance of \textsc{List Coloring}.
\end{proof}

\listcolalccqlem*

\noindent 
\begin{proof} {\bf (Continued)}
  We argue that the database $\Dmc$ constructed in the main part of
  the paper has cliquewidth~3. Since $G$ has cliquewidth~2, we can
  start with the 2-expression for $G$ and manipulate it to obtain a
  3-expression for $\Dmc$. Every edge insertion from label $i$ to
  label $j$ is replaced by two operations $\alpha_{i, j}^r$ and
  $\alpha_{j, i}^r$.  Whenever a constant $a_v$ with label
  $i \in \{1, 2\}$ is introduced, we also introduce the constants
  $b_v$, $a_v^c$ and $b_v^c$ for $c \in L \setminus C_v$, insert the
  relevant $r$-facts, and then rename the labels such that all $b_v$,
  $a_v^c$ and $b_v^c$ end up with label $3$ and $a_v$ with label $i$.
  It can be verified that this can be achieved using only 3 labels.
  
\medskip
  
For the sake of completeness, we list the atoms of the $i$-th gadget
of CQ~$q$
displayed in Figure~\ref{fig:cq-lower}.
\begin{align*}
&\{r(x, x), r(y, y), B(x), B(y)\} \, \cup\\
&\{r(x, z_1^i), r(z_1^i, z_2^i), r(z_2^i, z_3^i), r(z_3^i, y) \mid 1 \leq i \leq \ell\}\, \cup\\
&\{r(x, z_4^i), r(z_4^i, z_5^i), r(z_5^i, z_6^i), r(z_6^i, y) \mid 1 \leq i \leq \ell\}\, \cup\\
&\{A_i^0(z_1^i), A(z_2^i), A_i^1(z_3^i), A_i^1(z_4^i), A(z_5^i), A_i^0(z_6^i) \mid 1 \leq i \leq \ell\}
\end{align*}
\\[2mm]
\emph{Claim.} $G$ and $C_v$ are a yes-instance for \textsc{List Coloring}
if and only if $\Omc, \Dmc \not \models q$.
\\[2mm]
The `if' direction was already proved in the main body of the paper.

\smallskip

For `only if', let $G$ and the $C_v$ be a yes-instance
for \textsc{List Coloring} and $f : V \rightarrow \mathbb{N}$ a coloring witnessing this.
We show that there is a model $\Imc$ of $\Omc$ and $\Dmc$ such that
$\Imc \not \models q$. We obtain $\Imc$ from $\Dmc$ by coloring every $a_v$ with the color $f(v)$,
encoded by unary facts of the form $A_i^j(a_v)$, and every $b_v$ (resp.\ $b_v^c$) with the anti-color
of $a_v$ (resp.\ $b_v^c$). It is easy to check that $\Imc$ is a model of $\Omc$ and $\Dmc$.
We argue that $\Imc \not \models q$. If there was a homomorphism $h$ from $q$ to $\Imc$,
then by construction of $q$ and of $\Imc$, we must have $h(x) = b$ for some constant
$b \in \mn{adom}(\Dmc)$ of the form $b_v$ or $b_v^c$.
Let us assume that $b=b_v$ for some $v \in V$, the case $b=b_v^c$ is analogous.
Let $i \in \{1, \ldots, \ell\}$. Note that $q$ contains both the atoms $r(x, z_1^i)$ and $r(x, z_4^i)$,
but $b$ has only one $r$-neighbour and an $r$-selfloop. Thus, $h(z_1^i), h(z_4^i) \in \{b_v, a_v\}$.
Since $q$ contains the atoms $A_i^0(z_1^i)$ and $A_i^1(z_4^i)$, but the ontology forces
the extensions of $A_i^0$ and $A_i^1$ to be disjoint, we even have $\{h(z_1^i), h(z_4^i)\} = \{b_v, a_v\}$.
Analogously, let $h(y) = b_u$ for some $u \in V$, and it can be argued that
$\{h(z_3^i), h(z_6^i)\} = \{b_u, a_u\}$. By construction of $q$, we can further see that
either $h(x) = h(z_1^i)$ and $h(y) = h(z_6^i)$, or $h(x) = h(z_4^i)$ and $h(y) = h(z_3^i)$.
In any case, it follows that $a_v$ and $a_u$ are adjacent in $\Dmc$ and that they agree on the $i$-th
bit of their color. Since $i$ was chosen arbitrarily, $a_v$ and $a_u$ are colored using the same color.
But contradicts the fact, that $f$ is a coloring for $G$, so $\Imc \not \models q$.
\end{proof}

\section{Proofs for Section~\ref{sect:btw}}

As in the proof of  
Theorem~\ref{thm:alciucq2exp}, we may assume w.l.o.g.\ that $Q$ is  
Boolean.  
 Assume that we are given as input an OMQ
$Q=(\Omc,q) \in (\text{GF}_2, \text{UCQ})$, a database $\Dmc_0$ of
treewidth~$k$, and a candidate answer $\bar c$.

A tree decomposition $T=(V,E, (B_v)_{v \in V})$ of
$\Dmc_0$ of width at most $2k+1$ can be computed  in time
\mbox{$2^{O(k)} \cdot |\mn{adom}(\Dmc_0)|$} 
\cite{korhonen2021single}. We may assume that $(V,E)$ is a \emph{directed} tree
by choosing a root $v_0 \in V$. For each $v \in V$, let $\Dmc_v$ be
the restriction of database $\Dmc_0$ to the constants that appear in
some bag $B_u$, $u$ a node in the subtree of $(V,E)$ rooted at $v$.
Our aim is to traverse the tree $(V,E)$ bottom-up, computing for each
$v \in V$ a representation of the models \Imc of the database $\Dmc_v$
and ontology \Omc, enriched with additional information about partial
homomorphisms from CQs in the UCQ $q$ to \Imc.

To support our later constructions, we extend the ontology~\Omc. Note
that a CQ $p(\bar x)$ of treewidth~1 and arity~0 or 1 can be viewed
as a GF$_2$-formula $\varphi_p(\bar x)$ in an obvious way by reusing
variables. For example, if
$p(x) = \{r(x,x),r(x,y),s(y,x),r(y,z),A(z)\}$, then
$$\varphi_p= 
r(x,x) \wedge \exists y \, ( r(x,y) \wedge s(y,x) \wedge \exists x \,
(r(y,x) \wedge A(x))).$$ We use $\mn{trees}(q)$ to denote the set of all
Boolean or unary CQs of treewidth~1 that can be obtained from a CQ in
$q$ by first dropping atoms, then taking a contraction, and then
potentially selecting a variable as the answer variable.

Introduce a fresh unary relation symbol $A_p$ for every
$p \in \mn{trees}(q)$.  Let the ontology $\Omc'_q$ be obtained by adding
to \Omc the following:
\begin{itemize}

\item for every unary $p(x) \in \mn{trees}(q)$, the sentence
  $\forall x \, (A_{p(x)}(x) \leftrightarrow \varphi_{p(x)}(x))$;

\item for every Boolean $p \in \mn{trees}(q)$, the sentence
  $\forall x \, (A_p(x) \leftrightarrow \varphi_p())$.
   
\end{itemize}
For $i \in \{0,1,2\}$, the set $\mn{cl}_i(\Omc)$ is defined exactly
as in the proof of Theorem~\ref{thm:gf2aq2exp}.  We next introduce, for
each relevant formula in $\Omc'_q$, a relation symbol that may serve as
an abbreviation for that formula, and extend $\Omc'_q$ further to the
ontology $\Omc_q$:
\begin{itemize}

\item for every $\vp \in \mn{cl}_0(\Omc'_q)$, add a fresh
  unary relation symbol $A_\vp$ and the sentence $\forall x \, (A_\vp(x)
  \rightarrow \vp)$;

\item for every $\vp(x) \in \mn{cl}_1(\Omc'_q)$, add a fresh unary
  relation symbol $A_\vp$ and the sentence
  $\forall x \, (A_\vp(x) \rightarrow \vp(x)$.

\item for every $\vp(x,y) \in \mn{cl}_2(\Omc'_q)$, add a fresh
  binary relation symbol $A_\vp$ and the sentence
  $\forall x \forall y \, (A_\vp(x,y) \rightarrow \vp(x,y)))$.
  
\end{itemize}
Note that $|\Omc_q| \in |\Omc| \cdot 2^{O(|q|)}$.

We now define tree-extended models, in analogy with what was
done in Section~\ref{sect:alciucq}, but tailored towards GF$_2$ in place of
\ALCI.
Call a model \Imc of $\Dmc_0$
\emph{tree-extended} if it satisfies the following conditions:
\begin{itemize}

\item the Gaifman graphs of $\Dmc_0$ and
  $\Imc|_{\mn{adom}(\Dmc_0)}$ are identical;\footnote{That is,
    $\{ \{ c,c'\} \mid r(c,c') \in \Dmc_0 \} =
    \{ \{ c,c' \} \mid (c,c') \in r^{\Imc|_{\mn{adom}(\Dmc_0)}}\}$.}

\item if \Imc is modified by setting
  $r^\Imc = r^\Imc \setminus (\mn{adom}(\Dmc_0) \times
  \mn{adom}(\Dmc_0))$ for all role names $r$, then the result is a
  disjoint union of interpretations of treewidth~1, 
  each interpretation containing exactly one constant from $\mn{adom}(\Dmc_0)$.

\end{itemize}
Note that an interpretation of treewidth~1 is simply an
interpretation whose Gaifman graph is a tree.
The proof of
the following lemma is analogous to that of Lemma~\ref{lem:treeextended1}.
\begin{lemma}
  \label{lem:treeextended2}
  If there is a model \Imc of $\Dmc_0$ and \Omc such that
  $\Imc \not \models q$, then there is such a model \Imc that
  is tree-extended.
\end{lemma}
Based on Lemma~\ref{lem:treeextended1}, we rewrite the UCQ $q$ into a
UCQ $\widehat q$ 
such that for any model $\Imc$ of $\Dmc_0$ and $\Omc_q$,
$q(\Imc)=\widehat q(\Imc|_{\mn{adom}(\Dmc_0)})$.
The construction parallels the one used in Section~\ref{sect:alciucq}.

Consider all CQs that can be obtained in the following way. Start with
a contraction $p$ of a CQ from~$q$, then choose a set of
variables $S \subseteq \mn{var}(p)$ such that 
$$p^-=p \setminus \{ r(x,y) \in p \mid x,y \in S \}$$ is a disjoint union of
CQs of treewidth~1, each of which contains at most one variable from
$S$ and at least one variable that is not from~$S$.  Include in
$\widehat q$ all CQs that can be obtained by extending $p|_S$ as
follows:
\begin{enumerate}

\item for every maximal connected component $p'$ of $p^-$ that
  contains a (unique) variable $x_0 \in S$, 
  add the atom $A_{p'(x_0)}(x_0)$;


\item for every maximal connected component $p'$ of $p^-$ that
  contains no variable from $S$, add the atom $A_{p'}(z)$ with $z$
  a fresh variable.
  

\end{enumerate}
%
%
The number of CQs in $\widehat q$ is single exponential in $|q|$ and
each CQ is of size at most $|q|$. 

The proof of the following lemma is analogous to that of
Lemma~\ref{lem:qhatworksalc}.  Details are omitted.
\begin{restatable}{lemma}{qhatworksgflem}
  \label{lem:qhatworksgf}
  For every tree-extended model \Imc of $\Dmc_0$ and $\Omc_q$,
  $\Imc \models q$ iff $\Imc|_{\mn{adom}(\Dmc_0)} \models \widehat q$.
\end{restatable}
If the original UCQ $q$ was actually an AQ $q =A(x)$ and we were given
the answer candidate $c \in \mn{adom}(\Dmc_0)$, then we have reduced
this to a Boolean CQ by introducing a fresh concept name $B$, adding
$B(c)$ to $\Dmc_0$ and use the query $\exists x A(x) \wedge B(x)$.  It
is easy to see that Lemma~\ref{lem:qhatworksgf} holds if $\Omc_q$ is
replaced with \Omc and $\widehat q$ with $\exists x A(x) \wedge
B(x)$. For $(\text{GF}_2,\text{AQ})$, we can thus avoid extending the
ontology \Omc and instead work with the original one. This will help
in attaining the improved upper bound in Theorem~\ref{thmtwone}.

A \emph{diagram for a database} \Dmc is a database
\mbox{$\Dmc' \supseteq \Dmc$} that satisfies the following
conditions:
\begin{enumerate}

\item $\mn{adom}(\Dmc')=\mn{adom}(\Dmc)$;

\item $\Dmc' \setminus \Dmc$ uses only relation symbols from
  $\Omc_q$;

\item for every $\vp \in \mn{cl}_0(\Omc'_q) \cup \mn{cl}_1(\Omc'_q)$ and every $c \in \mn{adom}(\Dmc)$, $A_\vp(c) \in \Dmc'$
  or  $A_{\neg \vp}(c) \in \Dmc'$;
  
\item for every $\vp(x,y) \in \mn{cl}_2(\Omc'_q)$ and all
  $c,c' \in \mn{adom}(\Dmc)$, $A_{\vp}(c,c') \in \Dmc'$ or
  $A_{\neg \vp}(c,c') \in \Dmc'$;

\item if $A_{B(x)}(c) \in \Dmc'$, then $B(c) \in \Dmc'$; 

\item if $A_{r(x,y)}(c,c') \in \Dmc'$, then $r(c,c') \in \Dmc'$; 

\item $\Dmc'$ is satisfiable w.r.t.\ $\Omc_q$.

\end{enumerate}
Note that $\Dmc'$ uniquely fixes all 1-types and all 2-types through
Points~3 and~4. This serves the purpose of `synchronizing' the models
of different bags in the tree decomposition. Points~5 and~6 make the
`positive parts' of the types visible to CQs.
 We may simply speak of a
 \emph{diagram} when $\Dmc$ is not important.

\smallskip
Let $p$ be a Boolean CQ and $\Imc$ an interpretation. A \emph{partial match of
  $p$ into} \Dmc is a function
$m:\mn{var}(p) \rightarrow \mn{adom}(\Dmc) \cup \{ +, - \}$ that is a
homomorphism from the restriction of $p$ to the set of variables
$\{x \in \mn{var}(p) \mid m(x) \notin \{+,-\} \}$ to \Dmc.
We will also speak of partial matches into an interpretation $\Imc$
by seeing $\Imc$ as a database.
Informally, partial matches describe partial homomorphisms with
\mbox{$m(x)={-}$} meaning that variable $x$ could not yet be matched
and $m(x)={+}$ meaning that variable $x$ has already been successfully
matched elsewhere.  Let $p$ be a CQ, $k>0$ and let $m_i$ be a partial match of $p$ into $\Dmc_i$
for $i \in \{1, \ldots, k\}$. A partial match $m$ of $p$ into $\Dmc$ is a \emph{progression} of
$m_1,\dots,m_k$ if the following conditions hold:
\begin{enumerate}
\item if $m_i(x)=c$, then $m(x) \in \{c,+\}$;
\item if $m_i(x)=+$, then $m(x)=+$;
\item if $m(x)=+$, then there exists exactly one $i$ such that $m_i(x) \in \{+\} \cup \mn{adom}(\Dmc_i)$;
\item if $m(x)=+$ and $m_i(x) \in \mn{adom}(\Dmc_i)$, then $m_i(y) \in \{+\} \cup \mn{adom}(\Dmc_i)$
for all $y$ that appear together with $x$ in an atom of $p$.
\end{enumerate}
A partial match $m$ of $p$ is \emph{complete}
if $m(x) \neq {-}$ for all $x \in \mn{var}(p)$. 

A \emph{model abstraction} is a pair $(\Mmc,M)$ with \Mmc a diagram
and $M$ a set of pairs $(p,m)$ with $p$ a CQ in $\widehat q$ and $m$ a
partial match of $p$ into \Mmc. Every $v \in V$ and model \Imc of
$\Dmc_v$ and $\Omc_q$ give rise to a model abstraction
$\gamma_{\Imc,v}=(\Mmc_{\Imc, v},M_{\Imc,v})$ where
\begin{itemize}

\item $\Mmc_{\Imc, v}$ is obtained by starting with $\Dmc_0|_{B_v}$ and then adding:
  \begin{itemize}
  \item for every $\vp \in \mn{cl}_0(\Omc'_q)$ with $\Imc \models
    \vp$ and every $c \in B_v$, the fact $A_\vp(c)$;
  \item for every $\vp \in \mn{cl}_1(\Omc'_q)$  and every $c \in B_v$ with $\Imc \models
    \vp(c)$, the fact $A_\vp(c)$;
  \item for every $\vp \in \mn{cl}_2(\Omc'_q)$  and all $c_1,c_2 \in B_v$ with $\Imc \models
    \vp(c_1,c_2)$, the fact $A_\vp(c_1,c_2)$;
\item for every $B(x)\in \mn{cl}_1(\Omc'_q)$  and every $c \in B_v$ with $\Imc \models
    B(c)$, the fact $B(c)$;
  \item for every $r(x,y) \in \mn{cl}_2(\Omc'_q)$  and all $c_1,c_2 \in B_v$ with $\Imc \models
    r(c_1,c_2)$, the fact $r(c_1,c_2)$.
  \end{itemize}
  
\item $M_{\Imc, v}$ contains a pair $(p,m_{h,v})$ for every CQ $p$ in $\widehat q$,
  every induced subquery $p'$ of $p$, and every homomorphism $h$ from $p'$
  to \Imc with $\mn{range}(h) \subseteq \mn{adom}(\Dmc_v)$ where $m_{h_v}$ is the partial match of $p$ into $\Imc$, defined by
  \begin{itemize}

  \item $m_{h,v}(x)=h(x)$ if $x \in \mn{var}(p')$ and $h(x) \in B_v$;

  \item $m_{h,v}(x) = {+}$ if $x \in \mn{var}(p')$ and $h(x) \notin B_v$; 

      \item $m_{h,v}(x) = {-}$ if $x \notin \mn{var}(p')$.

  \end{itemize}
\end{itemize}
With an induced subquery $p'$ of $p$, we mean a CQ $p'$ that can be
obtained from $p$ by choosing a non-empty set
$V \subseteq \mn{var}(q)$ and then removing all atoms that use
a variable which is not in $V$. 

The central idea of our algorithm is to proceed in a bottom-up way
over the tree decomposition $(V,E)$, computing
for each node $v \in V$ the set
$$\Theta(v)=\{ \gamma_{\Imc,v} \mid \Imc \text{ model of } \Dmc_v \text{
  and } \Omc_q \}.
$$
We do this as follows:
\begin{itemize}
\item 
For all leaves $v \in V$, $\Theta(v)$ contains the set of all pairs
$(\Mmc,M)$ with \Mmc a diagram for $\Dmc_0|_{B_v}$ and $M$ the
set of all pairs $(p,m)$ with $p$ a CQ in $\widehat q$ and $m$ a partial
match of $p$ into \Mmc and $m(x) \neq +$ for all $x \in \mn{var}(p)$.

\item For all  non-leaves $v \in V$ with successors $v_1,\dots,v_k$,
$\Theta(v)$ is the smallest set that contains, for all
$(\Mmc_1,M_1) \in \Theta(v_1), \dots, (\Mmc_k,M_k) \in \Theta(v_k)$,
all pairs $(\Mmc,M)$ such that 
\begin{enumerate}

\item \Mmc is a diagram for $\Dmc_0|_{B_v}$ such that for
  $1 \leq i \leq k$, $\Mmc|_{B_v \cap B_{v_i}} = \Mmc_i|_{B_v \cap
    B_{v_i}}$, and

\item $M$ contains the set of all pairs $(p,m)$ such that $p$ is a CQ
  in $\widehat q$, $m$ is a partial match of $p$ into \Mmc, and for
  some $(p,m_1)\in M_1,\dots,(p,m_k)\in M_k$, $m$ is a progression
  of $m_1,\dots,m_k$.

\end{enumerate}
\end{itemize}
The following lemma states that the algorithm computes the intended set
of model abstractions.
\begin{lemma}
\label{lem:twalgointended}
$\Theta(v)=\{ \gamma_{\Imc,v} \mid \Imc \text{ model of } \Dmc_v \text{
  and } \Omc_q \}$ for every $v \in V$.
\end{lemma}

\begin{proof}
We prove the statement by induction on the depth of the subtree rooted at $v$.
The induction start ($v$ is a leaf) follows easily from the definition of $\Theta$.
For the induction step, let $v \in V$ not be a leaf and let $v_1,\ldots,v_k$ be the children of $v$.

For the first direction of the claim, let $(\Mmc, M) \in \Theta(v)$. We need to show that
there is a model $\Imc$ of $\Dmc_v$ and $\Omc_q$ such that
$(\Mmc, M) = (\Mmc_{\Imc,v}, M_{\Imc, v})$.
Since $(\Mmc, M) \in \Theta(v)$, there are pairs
$(\Mmc_1,M_1) \in \Theta(v_1), \dots, (\Mmc_k,M_k) \in \Theta(v_k)$ that satisfy the two conditions.
By the induction hypothesis, for every $i \in \{1, \ldots, k\}$, there is a model
$\Imc_i$ of $\Dmc_{v_i}$
and $\Omc_q$ such that $(\Mmc_i, M_i) = (\Mmc_{\Imc_i, v_i}, M_{\Imc_i, v_i})$.
Define $\Imc$ to be the union of all $\Imc_i$ and $\Mmc$.
By the first condition of the construction of $\Theta$, all these models are compatible,
so $\Imc$ is again a model.

It remains to show that $M$ is the set of all pairs $(p, m_{h,v})$
with $p$ a CQ in $\widehat q$ and $h$ a partial homomorphism from $p$ into $\Imc$
with $\mn{range}(h) \subseteq \mn{adom}(\Dmc_v)$.
Let $(p, m) \in M$. By definition of $\Theta$, $m$ is a progression of $m_1, \ldots, m_k$ for some
$(p, m_1) \in M_1, \ldots, (p, m_k)\in M_k$. Each of the $m_i$
takes the form $m_{h_i, v_i}$ for some partial homomorphism $h_i$ from $p$ into $\Imc_i$.
Since $m$ is a match of $p$ into $\Mmc$, $m$ restricted to the set
$\{x \in \mn{var}(p) \mid m(x) \notin\{+, -\}\}$ is a homomorphism. Let $h'$ be this homomorphism.
By the definition of a progression, all the $h_i$ and $h'$ are compatible, and we can
define $h$ to be the union of all the $h_i$ and $h'$, and it can be shown that $m=m_{h,v}$.
Now consider any pair $(p, m_{h, v})$ with $p$ a CQ in $\widehat q$
and $h$ a partial homomorphism from $p$ into $\Imc$.
Let $h_i$ be the restriction of $h$ to $h^{-1}(\mn{adom}(\Dmc_{v_i}))$.
Clearly, each $h_i$ is a partial homomorphism from $p$ to $\Imc_i$.
By the induction hypothesis, we have $(p, m_{h_i,v_i}) \in M_i$ for all $i \in \{1, \ldots, k\}$.
Furthermore $m_{h, v}$ satisfies the definition of being a progression of the $m_{h_i,v_i}$.
Thus, $m_{h, v} \in M$.

For the other direction of the claim,
let $\Imc$ be a model of $\Dmc_v$ and $\Omc_q$, and $M$ the set of all pairs $(p, m_{h,v})$
with $p$ a CQ in $\widehat q$ and $h$ a partial homomorphism from $p$ into $\Imc$.
We need to show that $(\Mmc_{\Imc, v}, M) \in \Theta(v)$.
Since $\Dmc_v \supseteq \Dmc_{v_i}$ for every $i \in \{1, \ldots, k\}$, $\Imc$
is a model of $\Dmc_i$ and $\Omc_q$ for every $i$.
Let $M_i$ be the set of all pairs $(p, m_{h, v_i})$ with $p$ a CQ in $\widehat q$
and $h$ a partial homomorphism from $p$ into $\Imc$.
By the induction hypothesis, we have $(\Mmc_{\Imc, v_i}, M_i) \in \Theta(v_i)$ for every
$i \in \{1, \ldots, k\}$.
Since $\Mmc_{\Imc, v}$ is an abstract model that is compatible with all the $\Mmc_{\Imc, v_i}$,
the sequence of pairs $(\Mmc_{\Imc, v_i}, M_i) \in \Theta(v_i)$ yields,
by definition of $\Theta$,
a pair $(\Mmc_{\Imc, v_i}, M') \in \Theta(v)$.
It remains to show that $M' = M$, but the proof is analogous to the proof in the previous direction.
\end{proof}

 The following lemma tells us how to answer once we have finished
our bottom-up traversal. It may be verified that the algorithm runs
within the intended time bounds.
\begin{restatable}{lemma}{twalgocorrectnesslem}
  \label{lem:tw-algo-correctness}
  Let $v_0$ be the root of $(V,E)$. Then $\Dmc_0 \not \models Q$ iff
  $\Theta(v_0)$ contains a pair $(\Mmc,M)$ such that $M$ contains
  no pair $(p,m)$ with $m$ a complete partial match.
\end{restatable}
\begin{proof}
First assume that $\Omc, \Dmc_0 \not \models q$.
Then there is a model $\Imc$ of $\Dmc_0$ and $\Omc$
such that $\Imc \not \models q$.
By Lemma~\ref{lem:treeextended2}, we can assume
that $\Imc$ is tree-extended.
We can extend $\Imc$ to a model of $\Omc_q$ in a straightforward way,
let this model also be called $\Imc$.
By Lemma~\ref{lem:qhatworksgf}, we have
$\Imc|_{\mn{adom}(\Dmc_0)} \not \models \widehat q$, so there is no homomorphism $h$
from any disjunct $p$ of $\widehat q$ to $\Imc|_{\mn{adom}(\Dmc_0)}$.
Lemma~\ref{lem:twalgointended} yields $(\Mmc_{\Imc, v_0}, M_{\Imc, v_0}) \in \Theta(v_0)$,
where $M_{\Imc, v_0}$ contains no complete partial match.
For the other direction, assume that $\Theta(v_0)$ contains a pair $(\Mmc, M)$ such that
$M$ contains no pair $(p, m)$ with $m$ a complete partial match.
By the claim, there is a model $\Imc$ of $\Dmc_0$ and $\Omc_q$ such that
$(\Mmc, M) = (\Mmc_{\Imc, v_0}, M_{\Imc, v_0})$.
If $\Imc \models q$, this would yield a pair $(p, m) \in M_{\Imc, v_0}$ with $m$ a complete partial match,
and thus, $\Imc \not \models q$, which implies $\Dmc_0, \Omc \not \models q$.
\end{proof}

\thmtwone*

\noindent 
\begin{proof}
Correctness of the given algorithm follows from Lemma~\ref{lem:tw-algo-correctness}.
It takes time $2^{O(k)} \cdot |\Dmc|$ to compute a tree-decomposition
of width $O(k)$.
The number of diagrams for a database of size $k$ is $2^{O(|\Omc|\cdot k^2)}$.
By the note after Lemma~\ref{lem:qhatworksgf}, $\widehat q$ contains only one disjunct.
There are at most $k+2$ many different partial matches, so there are
only $2^{k+2}$ many different sets $M$ of partial matches.
Thus, each $\Theta(v)$ contains at most
$2^{O(|\Omc|\cdot k^2)} \cdot 2^{k+2} = 2^{O(|\Omc|\cdot k)}$ many pairs.
This yields an overall running time of $2^{O(|\Omc|\cdot k^2)} \cdot |\Dmc|$,
which is $2^{O(|Q|\cdot k^2)} \cdot |\Dmc|$.
\end{proof}

\thmtwtwo*

\noindent 
\begin{proof}
  This proof is similar to the proof of Theorem~\ref{thmtwone}.  The
  only difference is that we have a UCQ instead of an AQ.  This yields
  at most $2^{O(|q| \log(|q|))}$ disjuncts in $\widehat q$, each of
  size at most $|q|$, and for every disjunct there are at most
  $(k+2)^{|q|}$ many possible partial matches into a database of size
  $k$.  Thus the number of possible sets $M$ of partial
  matches can be upper bounded by
  $2^{|\Omc| \cdot k^{O(|q| \log(|q|))}}$. 
  The number
  of diagrams for a database of size $k$ is still
  $2^{O(|\Omc|\cdot k^2)}$.  Thus, $\Theta(v)$ contains at most
$2^{|\Omc| \cdot k^{O(|q| \log(|q|))}}$
pairs.  This yields an overall running time of
$2^{|\Omc| \cdot k^{O(|q| \log(|q|))}} \cdot 2^{O(|q| \log(|q|))}$.
\end{proof}

\subsection{\ALCF with AQs}

A basic extension of \ALC with counting capabilities is the extension
\ALCF of \ALC with \emph{functional roles}, basic in comparison to
other extensions with counting such as (qualified) number
restrictions.  However, we prove that already evaluating OMQs from
$(\ALCF, \text{AQ})$ on databases of bounded treewidth is \coNP-hard,
provided that the unique name assumption (UNA) is not made. If it is made,
then the constructions provided here do not work and we conjecture
that all the results obtained for \ALC in the main body of the paper
still hold.

Let us first make precise what we mean with the UNA not being made.
In an interpretation \Imc, the interpretation function $\cdot^\Imc$
now additionally assigns an element $c^\Imc \in \Delta^\Imc$ to each
constant $a$. A database fact $A(c)$ is satisfied if
$c^\Imc \in A^\Imc$ and a fact $r(c,c')$ is satisfied if
$(c,c') \in r^\Imc$. Note that it is possible that
$c^\Imc = {c'}^\Imc$ despite $c \neq c'$, and it is in this sense that
names are not unique. The UNA may or may not be made in description
logic, both settings have received significant attention. Also note
that not making the UNA is the standard way to interpret constants
in classical first-order logic.

An \ALCF ontology is an \ALC ontology that can also contain
\emph{functionality assertions}, that is, assertions of the form
$\mn{func}(r)$ with $r$ a role name. An interpretation \Imc satisfies
$\mn{func}(r)$ if $r^\Imc$ is a partial function.

\thmalcf*

For the proof, we build on an observation due to Figueira. For
$n \geq 0$, let $\Dmc_{n \times n}$ denote the $n \times n$ \emph{grid
  database}, defined in the obvious way with role name $r_x$ used for
horizontal edges and $r_y$ for vertical edges. It is noted in
\cite{figueira2016semantically} that for one can find a database
$\widehat\Dmc_{n \times n}$ of treewidth~2 such that the models of
$\widehat\Dmc_{n \times n}$ and $\{\mn{func}(r_x)\}$ are precisely the
models of $\Dmc_{n \times n}$ and the size of $\Dmc_n$ is polynomial
in $n$.  We display the database $\widehat\Dmc_{n \times n}$ in
Figure~\ref{f:grid_Abox_to_constant_treewidth_general}. Note that,
when the UNA is made, that database is simply unsatisfiable w.r.t.\
$\{\mn{func}(r_x)\}$.

\tikzmath{\dist=0.1;}

\begin{figure}[ht]
  \begin{tikzpicture}[->,shorten >=0.5pt, >=latex,
    el/.style = {auto, scale=0.6, inner sep=0.7pt},
    every node/.style={inner sep=2pt}]
    \node [circle, draw] (1) {};
    \node [circle, draw] [right = of 1] (2) {};
    \node [circle, draw] [right = of 2] (40) {};
    \node [draw=none] [right = of 40] (3) {$\ldots\,$};
    \node [circle, draw] [right = of 3] (4) {};
    \node [circle, draw] [above = of 1] (5) {};
    \node [draw=none] at (5 -| 2) (h1) {}; 
    \node [circle, draw] [below = \dist em of h1] (6) {};
    \node [circle, draw] [right = of 6] (41) {};
    \node [draw=none] [right = of 41] (7) {$\ldots\,$};
    \node [circle, draw] [right = of 7] (8) {};
    \node [circle, draw] [above = \dist em of h1] (9) {};
    \node [circle, draw] [right = of 9] (42) {};
    \node [draw=none] [right = of 42] (10) {$\ldots\,$};
    \node [circle, draw] [right = of 10] (11) {};
    \node [circle, draw] [above = of 5] (12) {};
    \node [draw=none] at (12 -| h1) (h2) {}; 
    \node [circle, draw] [below = \dist em of h2] (13) {};
    \node [circle, draw] [right = of 13] (43) {};
    \node [draw=none] [right = of 43] (14) {$\ldots\,$};
    \node [circle, draw] [right = of 14] (15) {};
    \node [circle, draw] [above = of 12] (16) {};
    \node [draw=none] at (16 -| h1) (h3) {}; 
    \node [circle, draw] [above = \dist em of h3] (17) {};
    \node [circle, draw] [right = of 17] (44) {};
    \node [draw=none] [right = of 44] (18) {$\ldots\,$};
    \node [circle, draw] [right = of 18] (19) {};
    \node [circle, draw] [above = of 16] (20) {};
    \node [draw=none] at (20 -| h1) (h4) {}; 
    \node [circle, draw] [below = \dist em of h4] (21) {};
    \node [circle, draw] [right = of 21] (45) {};
    \node [draw=none] [right = of 45] (22) {$\ldots\,$};
    \node [circle, draw] [right = of 22] (23) {};
    \node [circle, draw] [above = \dist em of h4] (24) {};
    \node [circle, draw] [right = of 24] (46) {};
    \node [draw=none] [right = of 46] (25) {$\ldots\,$};
    \node [circle, draw] [right = of 25] (26) {};
    \node [circle, draw] [above = of 20] (27) {};
    \node [circle, draw] [right = of 27] (28) {};
    \node [circle, draw] [right = of 28] (47) {};
    \node [draw=none] [right = of 47] (29) {$\ldots\,$};
    \node [circle, draw] [right = of 29] (30) {};
    \draw [-, decorate, decoration = {brace,amplitude=10pt}] (-0.3,0) --  (-0.3,6) 
    node [midway, xshift=-0.55cm, outer sep=10pt,font=\footnotesize] {$n$};
    \draw [-, decorate, decoration = {brace,mirror,amplitude=10pt}] (0,-0.3) --  (5.3,-0.3) 
    node [midway, yshift=-0.55cm, outer sep=10pt,font=\footnotesize] {$n$};

    \path
      (1) edge node [el] {$r_x$}(2)
      (2) edge node [el] {$r_x$}(40)
      (40) edge node [el] {$r_x$}(3)
      (3) edge node [el] {$r_x$}(4)
      (5) edge node [el] {$r_x$}(6)
      (6) edge node [el] {$r_x$}(41)
      (41) edge node [el] {$r_x$}(7)
      (7) edge node [el] {$r_x$}(8)
      (5) edge node [el] {$r_x$}(9)
      (9) edge node [el] {$r_x$}(42)
      (42) edge node [el] {$r_x$}(10)
      (10) edge node [el] {$r_x$}(11)
      (12) edge node [el] {$r_x$}(13)
      (13) edge node [el] {$r_x$}(43)
      (43) edge node [el] {$r_x$}(14)
      (14) edge node [el] {$r_x$}(15)
      (16) edge node [el] {$r_x$}(17)
      (17) edge node [el] {$r_x$}(44)
      (44) edge node [el] {$r_x$}(18)
      (18) edge node [el] {$r_x$}(19)
      (20) edge node [el] {$r_x$}(21)
      (21) edge node [el] {$r_x$}(45)
      (45) edge node [el] {$r_x$}(22)
      (22) edge node [el] {$r_x$}(23)
      (20) edge node [el] {$r_x$}(24)
      (24) edge node [el] {$r_x$}(46)
      (46) edge node [el] {$r_x$}(25)
      (25) edge node [el] {$r_x$}(26)
      (27) edge node [el] {$r_x$}(28)
      (28) edge node [el] {$r_x$}(47)
      (47) edge node [el] {$r_x$}(29)
      (29) edge node [el] {$r_x$}(30)
    ;
    \path
      (1) edge node [el] {$r_y$}(5)
      (2) edge node [el] {$r_y$}(6)
      (4) edge node [el] {$r_y$}(8)
      (5) edge node [el] {$r_y$}(12)
      (9) edge node [el] {$r_y$}(13)
      (11) edge node [el] {$r_y$}(15)
      (12) -- node[pos=0.6] {$\vdots$}(16)
      (16) edge node [el] {$r_y$}(20)
      (17) edge node [el] {$r_y$}(21)
      (19) edge node [el] {$r_y$}(23)
      (20) edge node [el] {$r_y$}(27)
      (24) edge node [el] {$r_y$}(28)
      (26) edge node [el] {$r_y$}(30)
      (40) edge node [el] {$r_y$}(41)
      (42) edge node [el] {$r_y$}(43)
      (44) edge node [el] {$r_y$}(45)
      (46) edge node [el] {$r_y$}(47)
    ;
  \end{tikzpicture}
  \caption{The database $\widehat\Dmc_{n \times n}$}
  \label{f:grid_Abox_to_constant_treewidth_general}
\end{figure}

We prove Theorem~\ref{thm:alcf} by reduction from the word problem of
non-deterministic polynomially time-bounded Turing machines.  Such a
machine is a tuple
$M = (Q, \Sigma, \Gamma, \Delta, q_0, \blank, q_a,q_r)$ whose
components have the usual meaning; in particular, $\blank \in \Gamma
\setminus \Sigma$ is the blank symbol. We only mention that we work with
one-side infinite tapes, the transition relation $\Delta$ has the form
$\Delta \subseteq Q \times \Gamma \times Q \times \Gamma \times \{
L,R\}$, and $q_a$ and $q_r$ are the accepting and rejecting states
after which no further transitions are possible. We assume w.l.o.g.\
that $M$ never attempts to move left on the left-most tape cell.

Let $M$ be a $\mn{poly}$-time bounded NTM that solves an \NPclass-hard
problem. We provide a polynomial time reduction of
the word problem for $M$ to the complement of OMQ evaluation in
$(\ALCF,\text{AQ})$ on databases of treewidth~$2$.  

Let the input to $M$ be a word $w = a_1 \cdots a_n \in \Sigma^*$ and
consider the grid database $\Dmc_{m \times m}$ with $m=\mn{poly}(n)$. It is
not difficult to simulate the computation of $M$ on $w$ using the
evaluation of an OMQ from $(\ALCF,\text{AQ})$ on $\Dmc_{m \times m}$.
Moreover, due to the properties of $\widehat\Dmc_{m \times m}$
summarized above, in the reduction we may use the treewidth 2 database
$\widehat\Dmc_{m \times m}$ in place of $\Dmc_{m \times m}$. In the
following, we spell out the details.

We construct a database $\Dmc_{w}$ by starting with
$\widehat \Dmc_{m \times m}$. Let us assume that the constants in the
  first row are $c_1, \dots, c_{m}$.  We represent the initial
  configuration of $M$ by adding the following facts:
  \begin{itemize}
    \item $A_{q_0}(c_1)$ for the initial state $q_0$;
    \item $A_{a_1}(c_1), \dots, A_{a_n}(c_n), A_{\blank}(c_{n+1}), \dots, A_{\blank}(c_{\mn{poly}(|w|)})$ for the initial content of the tape
  \end{itemize}
  Next, we define an  ontology $\Omc_{M}$. It contains func$(r_x)$
  as well as the following CIs:
  \begin{enumerate}
    \item Put marker when choosing transition:
      $$
      A_q \sqcap A_a \sqsubseteq \forall r_y.B_{q_1, b_1, M_1} \sqcup
      \dots \sqcup B_{q_k, b_k, M_k}
      $$
      for all $q \in Q$ and $a \in \Gamma$ such that the tuples in
      $\Delta$ of the form $(q,a,\cdot,\cdot,\cdot)$ have last
      components $(q_1, b_1, M_1), \dots, (q_k, b_k, M_k)$

    \item Convert state transition markers to configuration:
      $$
      B_{q, b, M} \sqsubseteq A_b, \quad
      \forall r_x.B_{q, b, L} \sqsubseteq A_q, \quad
      B_{q, b, R} \sqsubseteq \forall r_x.A_q
    $$
     for all $q \in Q$, $b \in \Gamma$, and $M \in \{L, R\}$

    \item Check if rejecting state was reached and pass back that information:
    $$
      A_{q_r} \sqsubseteq F, \quad \exists r_x.F \sqsubseteq F, \quad \exists r_y.F \sqsubseteq F
    $$
  \item Mark cells that are not under the head:
    $$
        A_q \sqsubseteq \forall r_x.H_{\shortleftarrow} \qquad
        \forall r_x.A_q \sqsubseteq H_{\shortrightarrow}
    $$
for all $q \in Q$, as well as 
$$
      H_{\shortleftarrow} \sqsubseteq \forall r_x.H_{\shortleftarrow}, 
      \quad \forall r_x.H_{\shortrightarrow} \sqsubseteq H_{\shortrightarrow},
      \quad H_\shortleftarrow \sqcup H_\shortrightarrow \sqsubseteq N
$$
    \item Cells not under the head do not change:
    $$
      N \sqcap A_a \sqsubseteq \forall r_y.A_a
      \text{ for all $a \in \Gamma$}
    $$
    \item State, content of tape and head position are unique:
      $$
        A_q \sqcap A_{q'} \sqsubseteq \bot,
        \quad A_a \sqcap A_{a'} \sqsubseteq \bot,
        \quad N \sqcap A_q \sqsubseteq \bot
        $$
      for all $q, q' \in Q$ and $a,a' \in \Gamma$ with $q \neq q'$ and $a \neq a'$.
  \end{enumerate}
  As the AQ, we use $F(x)$, and thus our OMQ is $Q_M=(\Omc_m,F(x))$.
It is straightforward to show the following.
  \begin{lemma}
    $\Dmc_{w} \models Q_M(c_1)$ if and only if $M$ rejects $w$.
  \end{lemma}

\end{document}